\documentclass[11pt,a4paper]{article}

\usepackage{fancyhdr}
\usepackage{amsmath}
\usepackage{bm}
\usepackage{esint}
\usepackage{amsfonts}
\usepackage{amsthm}
\usepackage{amssymb}
\usepackage{amsbsy}
\usepackage{verbatim}
\usepackage{graphicx}
\usepackage{url}
\usepackage{mathrsfs}
\usepackage[hmargin = 1in,vmargin=.75in]{geometry}
\usepackage{color}
\usepackage{amstext}
\usepackage{stmaryrd}
\usepackage{enumerate}
\usepackage{algpseudocode}
\usepackage{algorithm}
\usepackage{pifont}
\usepackage{prettyref}
\font\msbm=msbm10

\numberwithin{equation}{section}

\theoremstyle{plain}
\newtheorem{theorem}{Theorem}[section]
\newtheorem{lemma}[theorem]{Lemma}
\newtheorem{claim}[theorem]{Claim}
\newtheorem{corollary}[theorem]{Corollary}

\newtheorem{definition}{Definition}[section]

\newtheorem{remark}[theorem]{Remark}
\def\mathbb#1{\hbox{\msbm{#1}}}

\newcommand{\be}{\boldsymbol{e}}

\newcommand{\bone}{\boldsymbol{1}}

\newcommand{\BA}{\boldsymbol{A}}

\newcommand{\BE}{\boldsymbol{E}}

\newcommand{\BG}{\boldsymbol{G}}

\newcommand{\BJ}{\boldsymbol{J}}

\newcommand{\BM}{\boldsymbol{M}}

\newcommand{\BP}{\boldsymbol{P}}
\newcommand{\BQ}{\boldsymbol{Q}}
\newcommand{\BR}{\boldsymbol{R}}
\newcommand{\BS}{\boldsymbol{S}}

\newcommand{\BU}{\boldsymbol{U}}
\newcommand{\BV}{\boldsymbol{V}}
\newcommand{\BW}{\boldsymbol{W}}
\newcommand{\BX}{\boldsymbol{X}}
\newcommand{\BY}{\boldsymbol{Y}}
\newcommand{\BZ}{\boldsymbol{Z}}

\newcommand{\BDelta}{\boldsymbol{\Delta}}

\newcommand{\BPhi}{\boldsymbol{\Phi}}
\newcommand{\BPsi}{\boldsymbol{\Psi}}

\newcommand{\BLambda}{\boldsymbol{\Lambda}}
\newcommand{\BSigma}{\boldsymbol{\Sigma}}

\newcommand{\PP}{\mathcal{P}}

\newcommand{\I}{\boldsymbol{I}}
\newcommand{\RR}{\mathbb{R}}
\newcommand{\lag}{\langle}
\newcommand{\rag}{\rangle}

\DeclareMathOperator{\argmin}{argmin}
\DeclareMathOperator{\argmax}{argmax}

\DeclareMathOperator{\E}{\mathbb{E}}

\DeclareMathOperator{\diag}{diag}

\newcommand{\Od}{\text{O}(d)}

\DeclareMathOperator{\rank}{rank}

\renewcommand{\Pr}{\mathbb{P}}

\long\def\\#1//{}

\begin{document}
\title{\bf Near-Optimal Performance Bounds for Orthogonal and Permutation Group Synchronization\\ via Spectral Methods}
\author{Shuyang Ling\thanks{New York University Shanghai (Email: sl3635@nyu.edu). This work is (partially) financially supported by the National Key R\&D Program of China, Project Number 2021YFA1002800, National Natural Science Foundation of China (NSFC) No.12001372, Shanghai Municipal Education Commission (SMEC) via Grant 0920000112, and NYU Shanghai Boost Fund.
} 
}

\maketitle

\begin{abstract}
Group synchronization asks to recover group elements from their pairwise measurements. It has found numerous applications across various scientific disciplines. In this work, we focus on orthogonal and permutation group synchronization which are widely used in computer vision such as object matching and structure from motion. Among many available approaches, the spectral methods have enjoyed great popularity due to their efficiency and convenience.
We will study the performance guarantees of the spectral methods in solving these two synchronization problems by investigating how well the computed eigenvectors approximate each group element individually. We establish our theory by applying the recent popular~\emph{leave-one-out} technique and derive a~\emph{block-wise} performance bound for the recovery of each group element via eigenvectors. In particular, for orthogonal group synchronization, we obtain a near-optimal performance bound for the group recovery in presence of additive Gaussian noise. For permutation group synchronization under random corruption, we show that the widely-used two-step procedure (spectral method plus rounding) can recover all the group elements exactly if the SNR (signal-to-noise ratio) is close to the information theoretical limit. Our numerical experiments confirm our theory and indicate a sharp phase transition for the exact group recovery.

\end{abstract}


\section{Introduction} 
Suppose there are $n$ group elements $\{g_i\}_{i=1}^n\in{\cal G}$ and we observe their noisy pairwise measurements 
\begin{equation}\label{eq:gsync}
g_{ij} = g_i^{-1}g_j +w_{ij}, \quad (i,j)\in{\cal E}
\end{equation}
where $w_{ij}$ is the noise and ${\cal E}$ is the edge set of an underlying network.
How to recover these elements $g_i$ from the noisy observations $\{g_{ij}\}_{(i,j)\in {\cal E}}$? Depending on the specific group type, the group synchronization problem is widely used in many applications including computer vision~\cite{OVBS17,SHSS16}, robotics~\cite{LHBC19,RCBL19}, clock synchronization~\cite{GK06} and cryo-electron microscopy~\cite{S18}. 
In this paper, we will focus on the synchronization of the orthogonal and permutation group.

\paragraph{Orthogonal group synchronization:} The group ${\cal G}$ in~\eqref{eq:gsync} is the orthogonal group $\Od$, 
\begin{equation}%\label{def:pmmat}
\Od : = \{ \BR: \BR\in\RR^{d\times d}, \quad \BR^{\top}\BR = \BR\BR^{\top} = \I_d \}.
\end{equation}

The general $\Od$ synchronization includes $\mathbb{Z}_2$-synchronization ($d=1$), angular synchronization $(d=2)$, SO(3) synchronization as special cases~\cite{ABBS14,AKKSB12,S11}. It often arises in rotation estimation and structure-from-motion~\cite{AKKSB12,OVBS17}, and also plays a significant role in SLAM (simultaneous localization and mapping) in robotics~\cite{LHBC19,RCBL19}.

\paragraph{Permutation group synchronization:}  
The underlying group ${\cal G}$ in~\eqref{eq:gsync} becomes permutation group, which is represented by permutation matrices $\Pi_d$:
\begin{equation}\label{def:pmmat}
\Pi_d : = \{\BR\in\{0,1\}^{d\times d}: ~ \BR^{\top}\BR = \BR\BR^{\top} = \I_d \}.
\end{equation}

Essentially, the permutation group synchronization is a special case of the $\Od$ synchronization since $\Pi_d$ is a subgroup of $\Od.$
Permutation group is directly related to the multi-way matching problem (map synchronization) in computer vision. Suppose there are $n$ images of the same object and each of them has $d$ features. Given a set of partially known feature correspondence among these $n$ images, how to find the~\emph{all the correct pairwise bijection}? This matching problem is one of the core problems in image registration, structure from motion, and object matching problem~\cite{HG13,PKS13,SHSS16}. This multi-way matching problem can be reformulated as recovering a set of permutation matrices from their pairwise products where each bijection corresponds to a permutation matrix~\cite{PKS13}.

\vskip0.25cm

Note that every element $\BR$ in $\Od$ or $\Pi(d)$ satisfies $\BR^{-1} = \BR^{\top}$. Therefore, the general synchronization problem reduces to recovering $n$ group elements $\{\BG_i\}_{i=1}^n$ from its noisy measurements
\[
\BG_{ij} = \BG_i\BG_j^{\top} + \text{noise}, \quad 1\leq i,j\leq n
\]
where the edge set ${\cal E}$ is assumed to be a complete graph throughout this manuscript. 
From now on, we let $\BA_G := [\BG_{ij}]_{1\leq i,j\leq n}$  be the $nd\times nd$ data matrix.

Given its practical importance, many efforts have been taken to solve the group synchronization problem. 
In absence of noise, group synchronization is easily solvable by sequentially recovering the group elements. However, this sequential strategy no longer works in presence of noise since the noise will be amplified. One common approach is to find the least squares estimator. However, it is usually an NP-hard problem to obtain the least squares estimator exactly, even for the simplest group $\mathbb{Z}_2=\{1,-1\}.$ 
As a result, many optimization approaches, including convex relaxation and nonconvex methods, are developed to tackle various challenging scenarios. 
In this work, we will instead focus on the spectral methods for orthogonal/permutation group synchronization. There are several variants of spectral methods for $\Od$ and $\Pi(d)$ group synchronization which are based on the data matrix $\BA_G := [\BG_{ij}]_{1\leq i,j\leq n}$~\cite{PKS13,SHSS16} or its corresponding (normalized) connection Laplacian matrix~\cite{BSS13}. Here we will focus the spectral methods which begin with computing the top $d$ eigenvectors of the observed data $\BA_G$ and then approximate each group element by rounding all the $d\times d$ blocks of the eigenvectors. In particular, we will investigate its performance and answer the following questions:
\[
\emph{\text{When does the spectral method recover the underlying group elements?}}
\]
\[
\emph{\text{How does the performance depend on the noise?}}
\]

\subsection{Related works and our contribution}

Group synchronization has found many applications in signal processing, computer vision, and machine learning. Some prominent examples include community detection~\cite{ABBS14,BBV16}  ($\mathbb{Z}_2$ synchronization), joint alignment~\cite{CC18} (finite cyclic group $\mathbb{Z}_n$), angular synchronization~\cite{BBS17,S11,ZB18}, statistical ranking~\cite{Y12} and phase retrieval~\cite{IPSV20} (unitary group U(1)), object matching~\cite{HG13,PKS13} (permutation group), rotation estimation~\cite{AKKSB12} (SO(3) group), clock synchronization~\cite{GK06} (cyclic group on a finite interval), and simultaneous localization and mapping (SLAM) in robotics~\cite{RCBL19,LHBC19} (special Euclidean group SE$(d)$). {There have been many efforts on solving the group synchronization problem in different settings by using various approaches including optimization-based approach~\cite{ABBS14,CC18,LYS20,RCBL19,WS13,ZB18}, spectral methods~\cite{S11,SHSS16,DCT21,Y12}, and message-passing type methods~\cite{LS19,PWBM18b}.}
For general group synchronization, one important topic is to determine how the noise strength affects the performance of algorithms and solvability.
The fundamental recovery criterion for information recovery from pairwise measurements is studied~\cite{CSG16}. The accuracy and noise sensitivity of the spectral method for general compact groups are presented in~\cite{RG20}. 
From now on, we will briefly review the recent literatures on orthogonal and permutation group synchronization and highlight those works which motivate this work.

Orthogonal group synchronization is often considered in rotation estimation arising from computer vision and robotics. One of the most widely used approaches to tackle general $\Od$ synchronization is to find the least squares estimator. As pointed out before, it is often an NP-hard problem to find the least squares estimator since the objective function is usually highly nonconvex and even discrete in some cases. This poses a significant challenge to  practical implementation. One important idea to overcome this technical difficulty is to find appropriate relaxations which are solvable within polynomial time. 
Convex relaxation has proven to be a very powerful method~\cite{WS13,Z19,L20b,RCBL19}. However, the solution to the convex relaxation program is not necessarily equal to that of the original program, i.e., the tightness does not always hold. The study of the tightness of convex relaxation has been a research focus in orthogonal group synchronization.
In~\cite{WS13}, Wang and Singer investigated the semidefinite program (SDP) relaxation of the orthogonal group synchronization under random corruption and characterized the phase transition of group recovery from noisy measurements. The tightness of the SDP relaxation for angular synchronization, as a special case of $\Od$ synchronization, is studied in~\cite{BBS17} with a near-optimal performance bound on the signal-to-noise ratio introduced in the very inspiring work~\cite{ZB18}. 
Recent works~\cite{Z19,L20b} propose suboptimal deterministic conditions which guarantee the tightness of the SDP relaxation for general $\Od$ synchronization. 
A similar route of research can also be found for permutation group synchronization. Huang and Guibas studied the convex relaxation approach of the permutation group synchronization in~\cite{HG13} and provided theoretical guarantees for correct recovery. The work~\cite{CGH14} investigated exact and robust object matching via SDP relaxation under partially known similarity between objects and the performance bound is near-optimal up to a log-factor. 
Despite the usefulness of convex relaxation, it remains highly nontrivial to solve large-scale SDPs. In practice, efficient first-order gradient-based approaches are preferred such as Riemannian optimization~\cite{AMS09,B15,RCBL19},  the Burer-Monteiro factorization~\cite{BM05,MMMO17,BVB20}, and {iterative reweighing strategy~\cite{SLL20}}. The major issue of Riemannian optimization is the inherent nonconvexity of the objective function, which could potentially create local optima. Fortunately, we have seen a surge of research in exploring the provably convergent nonconvex methods in solving $\mathbb{Z}_2$ synchronization~\cite{BBV16,LXB19}, angular synchronization~\cite{B16,LYS17,ZB18}, {permutation group synchronization~\cite{SLL20}}, and $\Od$ or SO$(d)$ synchronization in~\cite{B15,MMMO17,L20b,LS19}.

The spectral method is another popular method in group synchronization which is extremely convenient to use~\cite{ARF16,CKS15,DCT21,PKS13,PKSS14,S11,SHSS16,Y12}. 
Singer studied the spectral methods for angular synchronization~\cite{S11} with performance guarantees derived from random matrix theory.
The spectral method is also used for point cloud registration~\cite{CKS15} which is closely related to orthogonal group synchronization, and rigid-motion synchronization~\cite{ARF16} in special Euclidean group SE(3). {A recent work~\cite{DCT21} applies the SVD-based spectral methods to recover the ranking of some real numbers from their subsampled pairwise noisy differences which is a synchronization problem over the real line (a non-compact group).}
The work~\cite{PKS13} proposed the spectral relaxation of permutation group synchronization and {derived an $\ell_2$-norm performance bound} with tools from Gaussian random matrix theory;~\cite{SHSS16} provided a very careful block-wise analysis of the spectral methods for permutation group synchronization with a fixed underlying general network.
In~\cite{RG20}, the authors derived the sharp asymptotic formula for the mean squared error between the top eigenvectors of $\BA_G = [\BG_{ij}]_{1\leq i,j\leq n}$ and the planted group elements which exhibits a sharp phase transition. While the results in~\cite{RG20} hold for general compact group synchronization (of course applies to $\Od$ synchronization here), the block-wise recovery error bound for each group element was not obtained in~\cite{RG20}.

To derive a performance bound for the spectral methods, it suffices to approximate how close the eigenvectors of $\BA_G : =  [\BG_{ij}]_{1\leq i,j\leq n}$ are to the hidden group elements. This is essentially the perturbation of the eigenvectors of a clean low-rank matrix corrupted by random noise. 
This topic has been popular in random matrix theory, studied in a series of works~\cite{KX16,OVW16,OVW18}. 
Naive theoretical guarantees for the spectral methods can be easily derived by using Davis-Kahan theorem~\cite{DK70,S98,V07} which leads to an error bound under $\ell_2$- or Frobenius norm. However, the obtained bound is far from optimal since it does not yield a bound for each group component. Instead, we are more interested in an~\emph{entrywise} or \emph{block-wise} analysis of eigenvectors~\cite{EBW18,BDS21} which is usually quite challenging.
{Recently, there is an increasing trend of research focusing on providing an entrywise analysis of the eigenvectors in several statistical models~\cite{AFWZ20,DCT21,DLS21,FWZ18,CFMW19,MWCC20,ZB18}. In particular, the~\emph{leave-one-out} technique} has been shown highly powerful in deriving near-optimal performance bounds in the examples such as $\mathbb{Z}_2$- and angular synchronization~\cite{AFWZ20}, spectral clustering for stochastic block model~\cite{AFWZ20,DLS21}, ranking problem~\cite{CFMW19}, and covariance estimation~\cite{FWZ18}. It is also used in analyzing the convergence of first-order gradient method in solving inverse problems arising from signal processing and machine learning~\cite{MWCC20,ZB18}. 
Our work has benefitted greatly from~\cite{AFWZ20,ZB18} which provide an entrywise analysis of eigenvectors and its application in $\mathbb{Z}_2$-synchronization, community detection under the stochastic block model, and matrix completion. 

Our contribution consists of several aspects: we first study the spectral methods for $\Od$ synchronization under Gaussian noise: namely first computing the top $d$ eigenvectors of $\BA_G$ and use them to estimate the group elements. We provide a block-wise near-optimal error bound for each group element (modulo a constant) which justifies the usefulness of the spectral methods in $\Od$ synchronization. This analysis can be regarded as a natural generalization from $\mathbb{Z}_2$-synchronization in~\cite{AFWZ20} and angular synchronization in~\cite{ZB18}.
Then we study the permutation group synchronization under uniform random corruption. We are interested in when the two-step approach, namely, eigenvectors followed by rounding procedure, can give the exact recovery of the planted permutation matrix. The derived bound is also nearly optimal in terms of information theoretical limits, and improves the bound in~\cite{SHSS16} and matches the bounds obtained via the SDP relaxation in~\cite{HG13,CGH14}.  It is well worth noting that~\cite{BGHH18} studies a more general setting of permutation group synchronization and provides a near-optimal performance bound for the spectral methods. However, the technical approach is quite different from ours. Our theory is developed by applying the recent popular~\emph{leave-one-out} technique. However, the~\emph{block-wise} analysis of eigenvectors requires additional technical treatments. Our work resolved one question raised in~\cite{AFWZ20} about the block-wise analysis of eigenvectors for matrices with row/column block-wise independence. {This framework is quite flexible and can be applied to other problems which require the block-wise analysis of eigenvectors.}

\subsection{Organization}
Section~\ref{s:prelim} introduces the mathematical models and the spectral methods for group synchronization. We present the main results in Section~\ref{s:main} with numerical experiments to support our theory in Section~\ref{s:numerics}. The proofs are provided in Section~\ref{s:proof}.

\subsection{Notation}
Given a matrix $\BX$, $\BX^{\top}$ is the transpose of $\BX$ and $\BX\succeq 0$ means $\BX$ is positive semidefinite. $\I_n$ is the $n\times n$ identity matrix, {$\BJ_n$ is the $n\times n$ ``1" matrix}, and $\bone_n$ is an $n\times 1$ ``1" vector. $\|\BX\|$ denotes the operator norm of $\BX$ and $\|\BX\|_F$ is the Frobenius norm. For two matrices $\BX$ and $\BY$, we denote $\BX\otimes\BY$ their Kronecker product, i.e., the $(i,j)$-block of $\BX\otimes \BY$ is $X_{ij}\BY$. {For a matrix $\BX$, we let $\sigma_{i}(\BX)$ and $\lambda_i(\BX)$ be the $i$th largest singular value and eigenvalue of $\BX$ respectively.}
For two nonnegative functions $f(n)$ and $g(n)$, we denote $f(n)\lesssim g(n)$ and $f(n) = O(g(n))$ if there exists an absolute positive constant $C$ such that $f(n)\leq C g(n)$ for all $n$. 

\section{Preliminaries}\label{s:prelim}

This paper will study two benchmark models of group synchronization {under additive noise and uniform corruption (multiplicative noise).} 
\begin{itemize}
\item Orthogonal group synchronization under additive Gaussian noise. The pairwise noisy measurement $\BG_{ij}$ is observed between $\BG_i$ and $\BG_j$, 
\begin{equation}
\BG_{ij} = \BG_i\BG_j^{\top} + \sigma \BW_{ij}\label{model:od}\tag{OD}
\end{equation}
where $\BG_i\in \Od$ and $\BW_{ij}\in\RR^{d\times d}$ is a Gaussian random matrix.

\item Permutation group synchronization under uniform random corruption. 
Consider 
\begin{equation}
\BG_{ij} = \begin{cases}
\BG_i\BG_j^{\top}, &\text{ with probability } p, \\
\BP_{ij}, & \text{with probability } 1-p,
\end{cases} \label{model:pm}\tag{PM}
\end{equation}
where $\{\BG_i\}_{i=1}^n$ are the hidden permutation matrices and $\BP_{ij}\in\RR^{d\times d}$ is an independent random permutation uniformly sampled from $d!$ permutation matrices. In other words, 
\begin{equation}
\BG_{ij}= 
\begin{cases}
X_{ij} \BG_i\BG_j^{\top} + (1-X_{ij})\BP_{ij}, & i\neq j, \\
\I_d, & i = j,
\end{cases} \label{model:pm2}\tag{PM2}
\end{equation}
where $X_{ij}\sim$Bernoulli($p$) is independent of $\BP_{ij}.$
\end{itemize}

For both models, our goal is to recover $\BG_i$ from the noisy measurements $\BG_{ij}.$ One common method is to find the least squares estimator by minimizing
\[
\min_{\BR_i\in \Od}~ \sum_{i,j} \|\BR_i\BR_j^{\top} - \BG_{ij}\|_F^2 
\]
whose global minimizer equals the global maximizer of the following generalized quadratic form:
\begin{equation}\label{eq:prog}
\max_{\BR_i\in \Od} ~\sum_{i,j} \left\lag \BG_{ij}, \BR_i\BR_j^{\top} \right\rag.
\end{equation}
However, it is in general NP-hard to find the global optimizer. Therefore, one wants to find an appropriate relaxation of~\eqref{eq:prog}. The idea of spectral relaxation uses a simple fact: by letting $\BR$ be an $nd\times d$ matrix whose $i$th block equals $\BR_i$, then~\eqref{eq:prog} is equivalent to
\[
\max_{\BR\in \Od^{\otimes n}} \lag \BA_G, \BR\BR^{\top}\rag
\]
where $\BA_G$ is an $nd\times nd$ symmetric matrix whose $(i,j)$-block is $\BG_{ij}$.
Note that all $\BR\in\Od^{\otimes n}$ satisfies $\BR^{\top}\BR = n\I_d$. The spectral method simply replaces the constraints $\BR\in \Od^{\otimes n}$ by $\BR^{\top}\BR = n\I_d$,
\begin{equation}\label{eq:spectra}
\max_{\BR\in \RR^{nd\times d}} \lag \BA_G, \BR\BR^{\top}\rag\quad \text{ subject to } \quad\BR^{\top}\BR = n\I_d
\end{equation}
whose global maximizer equals the top $d$ eigenvectors of $\BA_{G}.$

\vskip0.2cm

As a result, the spectral method is very convenient to use: simply compute the top $d$ eigenvectors of the matrix $\BA_G$, denoted by an $nd\times d$ partial orthogonal matrix $\BPhi$  where $\BPhi^{\top} = [\BPhi_1^{\top}, \cdots, \BPhi_n^{\top}]$ and $\BPhi_i$ is the $i$th $d\times d$ block. In particular, we normalize $\BPhi$ to be $\BPhi^{\top}\BPhi = n\I_d$, i.e.,  each column is of norm $\sqrt{n}$.
%Let $\BLambda\in\RR^{d\times d}$ be the diagonal matrix consisting of its top $d$ eigenvalues and we have $\BA_G\BPhi = \BPhi\BLambda$. 
Then we implement a rounding procedure to obtain the estimation of $\BG_i.$ We summarize the aforementioned procedures in Algorithm~\ref{algo1}.

\vskip0.25cm

\begin{algorithm}[h!]
\caption{Spectral methods for orthogonal group synchronization}
\begin{algorithmic}[1]
\State Compute the top $d$ eigenvectors $\BPhi$ of $\BA_G$ with $\BPhi^{\top}\BPhi = n\I_d.$
\State Compute $\widehat{\BG}_i = \PP(\BPhi_i)$ for all $1\leq i\leq n$ where $\BPhi_i$ is the $i$th block of $\BPhi$ and
\begin{equation}\label{def:P}
\PP(\BX) : = \BU\BV^{\top}
\end{equation}
where $\BU$ and $\BV$ are the left and right singular vectors of $\BX.$

\end{algorithmic}
\label{algo1}
\end{algorithm}

For permutation matrix, a slight modification of the rounding procedure is implemented. Simply speaking, once we get $\BPhi_i$, we estimate $\BG_i$ via 
\[
\widehat{\BG}_i = \argmin_{\BR_i \in \Pi_d } \left\| \BR_i - \PP(\BPhi_1)\PP(\BPhi_i)^{\top} \right\|_F^2
\]
where $\Pi_d$ is the set of all $d\times d$ permutation matrices.
This linear assignment problem can be solved by the Hungarian algorithm in polynomial time~\cite{K55}. The entire procedures are summarized in Algorithm~\ref{algo2}.

\begin{algorithm}[h!]
\caption{Spectral methods for permutation group synchronization}
\begin{algorithmic}[1]
\State Compute the top $d$ eigenvectors $\BPhi$ of $\BA_G$ with $\BPhi^{\top}\BPhi = n\I_d.$
\State Compute $\PP(\BPhi_i)$ for all $1\leq i\leq n$ where $\BPhi_i$ is the $i$th block of $\BPhi$ and
\[
\PP(\BX) : = \BU\BV^{\top}
\]
where $\BU$ and $\BV$ are the left and right singular vectors of $\BX.$

\item Compute $\PP(\BPhi_1)\PP(\BPhi_i)^{\top}$ for $1\leq i\leq n.$
\item Get the estimate for $\BG_i$ via
\[
\widehat{\BG}_i = \argmin_{\BR_i \in \Pi_d } \| \BR_i - \PP(\BPhi_1)\PP(\BPhi_i)^{\top} \|_F^2
\]
\end{algorithmic}
\label{algo2}
\end{algorithm}

How well do these algorithms work? {We consider the $\Od$ synchronization under additive Gaussian noise as an example. The data matrix $\BA_G$ can be naturally written into a spiked matrix model: $\BA_G = \BG\BG^{\top} + \BDelta$ where $\BDelta = \sigma\BW$. Note that without any noise, the top $d$ eigenvectors exactly give the group elements.}
If the noise $\BDelta$ is small, then one can easily invoke the classical matrix perturbation argument, e.g. Davis-Kahan theorem (Theorem~\ref{thm:dk}), to obtain an error bound between the top $d$ eigenvectors and the planted group elements in terms of operator or Frobenius norm. Namely, 
\[
{\min_{\BQ\in \Od}\|\BPhi - \BG \BQ\| \lesssim \frac{\|\BDelta\BG\|}{\lambda_{d}(\BA_G)} }
\]
which will be derived more carefully later in the proof section. 

On the other hand, it is much more appealing to provide an error bound for
\[
\max_{1\leq i\leq n} \| \widehat{\BG}_i - \BG_i\BQ \| \qquad \text{where }\qquad\widehat{\BG}_i = \PP(\BPhi_i)
\]
for some orthogonal matrix $\BQ\in\RR^{d\times d}$ since this would provide us an error bound for the recovery of each group element. In other words, we need to control the estimation error for each block $\BPhi_i$, which is essentially a generalization of the entrywise bound for the eigenvector discussed in~\cite{AFWZ20,ZB18}.
 However, the Davis-Kahan bound does not immediately yield a tight bound for the deviation of each $\BPhi_i$ from $\BG_i\BQ$  for some $\BQ\in\Od$. {This will be the main focus of our paper: we obtain the block-wise perturbation bound of $\BPhi$ via the~\emph{leave-one-out technique}. We will introduce this technique briefly in Section~\ref{ss:loo} and provide more details in Section~\ref{s:proof}.}

\section{Main theorem}\label{s:main}

In this section, we will provide theoretical guarantees for the spectral methods in solving the $\Od$ and $\Pi_d$ synchronization problem under the statistical models~\eqref{model:od} and~\eqref{model:pm} respectively.

\subsection{Main results}
Our main contribution is providing a near-optimal block-wise error bound of $\widehat{\BG}_i$ for all $1\leq i\leq n$. For the $\Od$ synchronization under Gaussian noise, we have the following theorem.

\begin{theorem}[Performance for orthogonal group synchronization]\label{thm:od}
Suppose the parameter $\sigma$ in the model~\eqref{model:od} satisfies
\[
\sigma <\frac{ c_0\sqrt{n}}{\sqrt{d} + \sqrt{\log n}}
\]
for some small constant $c_0 > 0.$ Then with high probability, the estimation $\widehat{\BG}_i$ of $\BG_i$ from Algorithm~\ref{algo1} satisfies
\[
\left\| \widehat{\BG}_i  \widehat{\BG}_j^{\top}-\BG_i\BG_j^{\top} \right\| \lesssim \sigma\sqrt{n^{-1}d}, \quad \forall i\neq j.
\]
In other words, 
\[
\max_{1\leq i\leq n}\|\widehat{\BG}_i - \BG_i\BQ_j\| \lesssim \sigma\sqrt{n^{-1}d}
\]
by letting $\BQ_j = \BG_j^{\top}\widehat{\BG}_j$ for any $1\leq j\leq n.$
\end{theorem}

\begin{remark}
Theorem~\ref{thm:od} includes $\mathbb{Z}_2$- and angular synchronization as special cases. In particular, if $d = 1$, the problem reduces to $\mathbb{Z}_2$-synchronization and the bound is equivalent to the one derived in~\cite{AFWZ20}; for $d=2$, our result is closely related to the angular synchronization explored in~\cite{ZB18} since SO(2) is isomorphic to U(1).
\end{remark}

{
Simply speaking, Theorem~\ref{thm:od} provides a theoretical guarantee for the spectral estimator in the orthogonal group synchronization under additive Gaussian noise: the distance of the spectral estimator from the planted signal is controlled by the noise strength. As discussed before, the spectral methods are viewed as a relaxation of the equivalent least squares objective function~\eqref{eq:prog}. Therefore, they are unlikely to produce the globally optimal least squares estimator~\eqref{eq:prog}. However, the proximity of the spectral estimator to the ground truth provides allows nonconvex optimization approaches to have a high-quality initialization and enjoy a global convergence to the globally optimal least squares estimator~\cite{CC18,MWCC20,ZB18,L20c,L21a,L21b}.}

{
Now we briefly discuss the optimality of our result. Note that the model $\BA_G = \BG\BG^{\top} + \sigma\BW$ for the $\Od$ synchronization under additive Gaussian is essentially the well-known spiked matrix model or the real deformed Wigner matrices~\cite{BN11,CDF09}.}
Note that in random matrix theory, it has been extensively studied when the top eigenvectors of $\BA_G$ are correlated with the planted signals (low-rank matrix), see e.g.~\cite{BN11,PWBM18,CDF09}. For this finite-rank spiked matrix model, it has been shown in~\cite{CDF09} if the noise level $\sigma$ is above the threshold $\sigma > \sqrt{n/d}$, the leading $d$ eigenvalues of $\BA_G$ fail to exit the limiting semicircle compact support of the GOE (Gaussian orthogonal ensemble) for a sufficiently large $n$. This implies the spectral method (plus rounding) is expected to identify the planted signal only in the regime $\sigma \lesssim \sqrt{n/d}$.
Thus our bound in Theorem~\ref{thm:od} differs from this threshold only by a logarithmic and constant factor. Though not explicitly stated, it is believed that $\sigma =\sqrt{n/d}$ is the threshold above which is  information-theoretically possible to detect the spikes~\cite{JCL20}.

\vskip0.25cm

The theoretical result for permutation group synchronization is summarized as follows. 
\begin{theorem}[Performance for permutation group synchronization]\label{thm:pm}
Suppose the parameter $p$ in the model~\eqref{model:pm} satisfies
\begin{equation}\label{eq:pm_bound}
p > C_0\sqrt{\frac{\log (nd)}{n}}
\end{equation}
for some universal large constant $C_0 > 0$.
Then with high probability, the estimation $\widehat{\BG}_i$ of $\BG_i$ from Algorithm~\ref{algo2} satisfies
\[
\left\| \widehat{\BG}_i  \widehat{\BG}_j^{\top}-\BG_i\BG_j^{\top} \right\| \lesssim \frac{1}{p}  \sqrt{\frac{\log (nd)}{n}}, \quad \forall i\neq j.
\]
In particular, if $\| \widehat{\BG}_i\widehat{\BG}_j^{\top}-\BG_i\BG_j^{\top} \|  < \frac{1}{2}$, then Algorithm~\ref{algo2} recovers the hidden permutation matrices $\BG_i$ exactly.

\end{theorem}

\begin{remark}
The work~\cite{SHSS16} provides a block-wise bound for permutation group synchronization on general networks in which $p > C_0 n^{-\frac{1}{2}}\log^3 (n)$ is needed for the exact recovery of all the permutation matrices with high probability.~\cite{CGH14} shows that the SDP relaxation can recover the underlying hidden permutation matrices with high probability if $p > C_0 n^{-\frac{1}{2}}\log^2(nd)$. The bound~\eqref{eq:pm_bound} matches the state-of-the-art performance bound in~\cite{BGHH18} which considers the general simultaneous mapping and clustering problem. However, as pointed out earlier, our technique is quite different from~\cite{BGHH18}. {Note that the information theoretic limit for the exact recovery in $\Pi_d$ synchronization is discussed in~\cite[Corollary 1]{CSG16}: no method whatsoever is able to recover the ground truth if $p < O(1/\sqrt{n})$. Therefore, our bound differs from the information-theoretic limit by a logarithmic factor.  }
\end{remark}

Another synchronization model which is highly relevant to the two aforementioned models is the $\Od$ group synchronization with uniform multiplicative noise~\cite{WS13}:
\[
\BG_{ij} = 
\begin{cases}
\BG_i\BG_j^{\top}, & \text{with probability }p, \\
\BR_{ij}, & \text{with probability }1-p, \\
\end{cases}
\]
where $\BR_{ij}$ is sampled from the uniform Haar distribution over $\Od$\footnote{Simply speaking, Haar distribution on $\Od$ is the unique invariant probability measure on the compact group $\Od$.}. Though it is not analyzed in our manuscript, the proof technique for the permutation group synchronization under uniform corruption could be directly modified to tackle this $\Od$ synchronization under uniform multiplicative corruption.

\subsection{The sketch of proof: leave-one-out technique}\label{ss:loo}
We provide a proof sketch for Theorem~\ref{thm:od} and~\ref{thm:pm}, and will proceed to give more technical details in Section~\ref{s:proof}. The main idea follows from the~\emph{leave-one-out} technique employed in~\cite{AFWZ20} to study $\mathbb{Z}_2$-synchronization and community detection under the stochastic block model. The major difference of our setting here is the blockwise independence of the noise matrix as well as the multi-dimensionality of the eigenspace, which requires additional technical treatments.

With a bit of calculation, both~\eqref{model:od} and~\eqref{model:pm} can be formulated under the framework of the spiked matrix model. Without loss of generality, we assume each $\BG_i$ is an identity matrix $\I_d$ and it suffices to consider
\begin{equation}\label{model:ms}
\BA := \BZ\BZ^{\top} + \BDelta
\end{equation}
where $\BZ^{\top} = [\I_d,\cdots,\I_d]\in\RR^{d\times nd}$. Here $\BDelta$ is the random noise matrix. More precisely, 
\begin{itemize}
\item For model~\eqref{model:od}, the noise matrix $\BDelta$ is
\[
\BDelta = \sigma \BW, 
\]
where $\BW\in\RR^{nd\times nd}$ is a symmetric Gaussian random matrix.

\item For model~\eqref{model:pm}, the corruption matrix $\BDelta$ is
\begin{equation}\label{def:Deltapm}
\BDelta_{ij} = 
\begin{cases}
-p^{-1}(1-p)d^{-1}\BJ_d, & \text{ if }i = j, \\
p^{-1}\left( (X_{ij} - p) (\I_d - d^{-1}\BJ_d) + (1-X_{ij})(\BP_{ij} - d^{-1}\BJ_d)\right),  & \text{ if } i\neq j.
\end{cases}
%\begin{cases}
%p^{-1}(1-p)d^{-1}\BJ_d, & \text{ if }i = j, \\
%p^{-1}\left( (X_{ij} - p) (\I_d - d^{-1}\BJ_d) + (1-X_{ij})(\BP_{ij} - d^{-1}\BJ_d)\right), & \text{ if } i\neq j.
%\end{cases}
\end{equation}
where $\BP_{ij}$ is a random permutation matrix drawn uniformly from the set of all $d\times d$ permutation matrices {and $\BJ_d$ is a $d\times d$ matrix whose entries are all equal to 1}. In fact, the top $d$ eigenvectors of $\BA$ and $\BA_G = [\BG_{ij}]_{1\leq i,j\leq n}$ are the same provided that the noise $\BDelta$ is small and $\BG_i = \I_d$. We will justify this fact in Lemma~\ref{lem:AG} in Section~\ref{s:proofpm}.

\end{itemize}

Now we briefly introduce the main idea of the leave-one-out technique in obtaining a block-wise error bound for the top $d$ eigenvectors $\BPhi$ of $\BA$. Assume $(\BPhi, \BLambda)$ is the top $d$ leading eigen-pairs of $\BA$, i.e., 
\[
\BA\BPhi = \BPhi\BLambda, \quad \BPhi^{\top}\BPhi = n\I_d, \quad \BLambda = \diag(\lambda_1, \cdots, \lambda_d).
\]
In other words, it holds
$\BPhi = \BA\BPhi\BLambda^{-1}.$ %Our goal is to show that all $d\times d$ blocks $\{\BPhi_i\}_{i=1}^n$ are very close, so are their ``generalized phase" $\PP(\BPhi_i)$ defined in~\eqref{def:P}. This is intuitively true since $\BZ$ is the planted signal in the noisy measurements $\BA = \BZ\BZ^{\top} + \BDelta$ and we believe the noise will not lead to a drastic perturbation for each block.
The idea of estimating each $\BPhi_i$ relies on choosing a suitable~\emph{surrogate} which is easy to approximate and also close to $\BPhi_i$. One commonly-used choice is to use one-step fixed point iteration, which is inspired by~\cite{AFWZ20}.
By definition, $\BPhi\in\RR^{nd\times d}$ is the fixed point of the following map:
\begin{equation}\label{eq:fix}
f(\BX) := \BA\BX\BLambda^{-1}
\end{equation}
where $\BLambda\in\RR^{d\times d}$ consists of the top $d$ eigenvectors of $\BA.$  

Note that the recovered orthogonal group $\{\BG_i\}_{i=1}^n$ is unique modulo a global rotation. Therefore, we initialize this fixed point map~\eqref{eq:fix} by
 choosing $\BX = \BZ\BQ$ where $\BQ$ minimizes the distance $d_F(\BPhi,\BZ)$ between $\BPhi$ and $\BZ$ is minimized, i.e., 
 \begin{equation}
d_F(\BPhi,\BZ) : =  \min_{\BQ\in \Od}\|\BPhi - \BZ\BQ\|_F.
 \end{equation}
We hope $f(\BZ\BQ)$ is close to $\BPhi$ \emph{uniformly} for each $d\times d$ block. Let's perform a preliminary analysis for the approximation error bound of $\BPhi_i$ with the $i$th block of $f(\BZ\BQ) = \BA\BZ\BQ\BLambda^{-1}.$ Let $\BDelta_i$ be the $i$th block column of $\BDelta$, and then it holds
\begin{align}
\| [\BPhi - \BA\BZ\BQ\BLambda^{-1}]_i \| & =\| [ \BA(\BPhi- \BZ\BQ)\BLambda^{-1}]_i \| \nonumber \\
& = \| (\BZ + \BDelta_i)^{\top} (\BPhi- \BZ\BQ)\BLambda^{-1} \| \nonumber \\
& \leq \|\BLambda^{-1}\|  \cdot \left(  \| \BZ^{\top}(\BPhi- \BZ\BQ) \| + \|  \BDelta_i^{\top} (\BPhi- \BZ\BQ) \| \right). \label{eq:keyest}
\end{align}

The Davis-Kahan theorem~\cite{DK70} provides a tight bound of the first term and the goal is to estimate $\|  \BDelta_i^{\top} (\BPhi- \BZ\BQ) \|$. Note that $\BDelta_{i}$ and $\BPhi-\BZ\BQ$ are not statistically independent. Therefore, despite that each $\BDelta_{ij}$ is either a Gaussian random matrix or a bounded centered random permutation matrix, we cannot immediately apply concentration inequality to obtain a tight bound of $\|  \BDelta_i^{\top} (\BPhi- \BZ\BQ) \|$. The remedy is to use the recently popular~\emph{leave-one-out} trick.

\vskip0.25cm

The idea is to replace $\BPhi$ by $\BPhi^{(i)}$ which is the top $d$ eigenvectors of the following auxiliary matrix $\BA^{(i)} = \BZ\BZ^{\top} + \BDelta^{(i)}$:
\begin{equation}\label{def:Deltai}
\BA^{(i)}_{k\ell} = 
\begin{cases}
\BA_{k\ell}, & \text{ if } k \neq i \text{ and } \ell \neq i, \\
\I_d, & \text{ if } k = i \text{ or } \ell= i,
\end{cases} \qquad 
\BDelta^{(i)} = 
\begin{cases}
\BDelta_{k\ell}, & \text{ if } k \neq i \text{ and } \ell \neq i, \\
0, & \text{ if } k = i \text{ or } \ell= i.
\end{cases}
\end{equation}
In other words, $\BA$ and $\BA^{(i)}$ only differ by the $i$th block column and row of $\BDelta.$ Because of this minor difference, the corresponding eigenspace $\BPhi$ and $\BPhi^{(i)}$ are very close. More importantly, $\BPhi^{(i)}$ is independent of $\BDelta_i$ since $\BPhi^{(i)}$ only depends on $\BDelta^{(i)}$ which excludes $\BDelta_i$. This important fact allows one to apply the concentration inequality to get a satisfactory bound of $\|\BDelta_i^{\top} (\BPhi - \BZ\BQ) \|$ which will be discussed in more details in Section~\ref{s:proof}.

\section{Numerics}\label{s:numerics}

In this section, we will provide numerical evidence to show that the bound in the main theorems are near-optimal.

\subsection{Orthogonal group synchronization}
We first investigate the performance of Algorithm 1 under various noise levels.
Consider $\BA_G = \BG\BG^{\top} + \sigma\BW$ where $\BW$ is a symmetric $nd\times nd$ Gaussian random matrix where $n=1000$ and $d=2,3,5$ and 10. Here we introduce another parameter  $\kappa$ such that $\sigma = \kappa\sqrt{n/d}$ because
Theorem~\ref{thm:od} implies
\begin{equation}\label{eq:num1}
\sigma = \kappa \sqrt{\frac{n}{d}}, \quad \kappa \leq 
\frac{c_0}{1 + \sqrt{d^{-1}\log n}}
\end{equation}
would provide a non-trivial bound for some constant $c_0.$
We let $\kappa$ vary from $0.05$ to 0.5 since we know that the spectral method is expected to fail for $\kappa > 1$ and to succeed for $\kappa <1$ based on the random matrix theory.
Here once we obtain $\widehat{\BG}_i$, we calculate the maximum blockwise deviation of $\widehat{\BG}_i$ from $\BG_i$ by using 
\[
\max_{ i\neq 1}\|\widehat{\BG}_1\widehat{\BG}_i^{\top} - \BG_1\BG_i^{\top}\|.
\]
For each $\kappa$, we run 25 experiments and obtain the boxplot of error, as is shown in Figure~\ref{fig:od_strong}. {The bottom/top edges of each blue box stand for the 25th and 75th percentile of the estimation error, the central red mark indicates the median error, and each red dot is an outlier. We can see that the median block-wise error grows approximately linearly with respect to $\kappa$ (equivalently, $\sigma$) if $\kappa$ is smaller than some threshold for each fixed $d$. Once $\kappa$ exceeds that threshold, the algorithm cannot provide a nontrivial estimation of $\BG_i$ as the error becomes 2. In addition, it is also interesting to see that the range of $\kappa$ for a nontrivial error bound (i.e., the threshold) gets larger as $d$ increases. This may be explained by the inequality above~\eqref{eq:num1}: as $d$ gets larger, the denominator $1+\sqrt{d^{-1}\log n}$ in~\eqref{eq:num1} becomes smaller and thus allows a larger range of $\kappa$ for a non-trivial error bound.}

\begin{figure}[h!]
\centering
\begin{minipage}{0.48\textwidth}
\includegraphics[width=80mm]{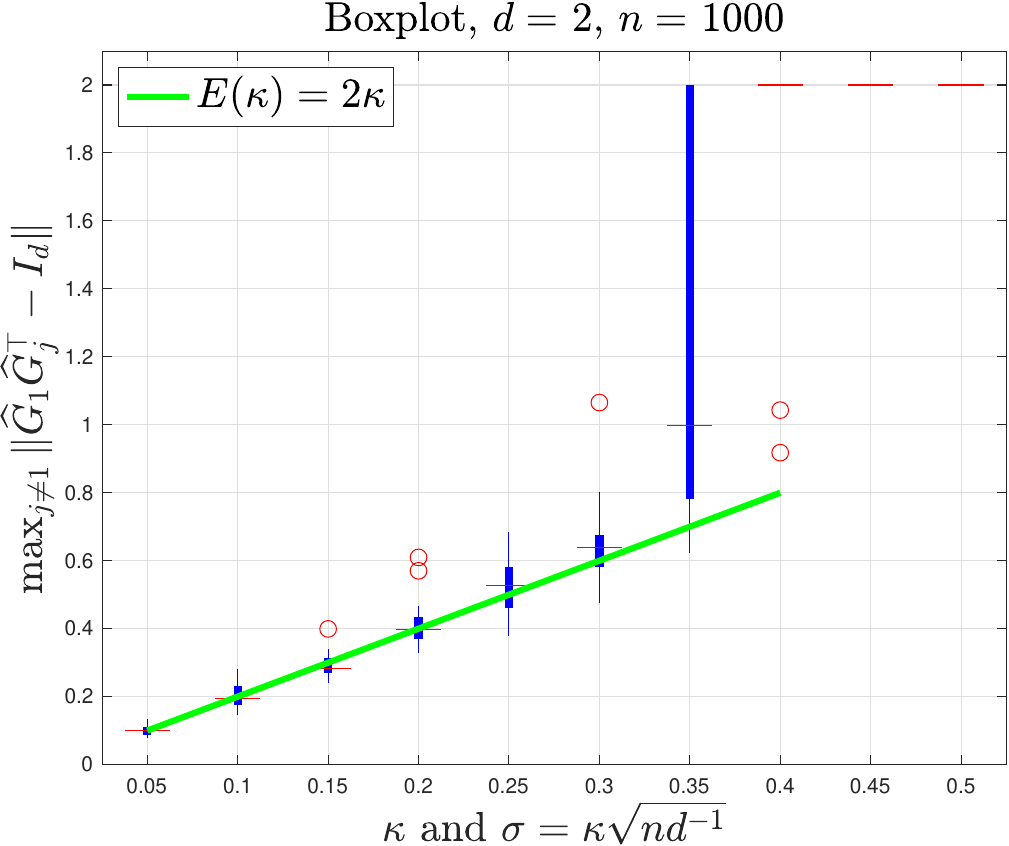}
\end{minipage}
\hfill
\begin{minipage}{0.48\textwidth}
\includegraphics[width=80mm]{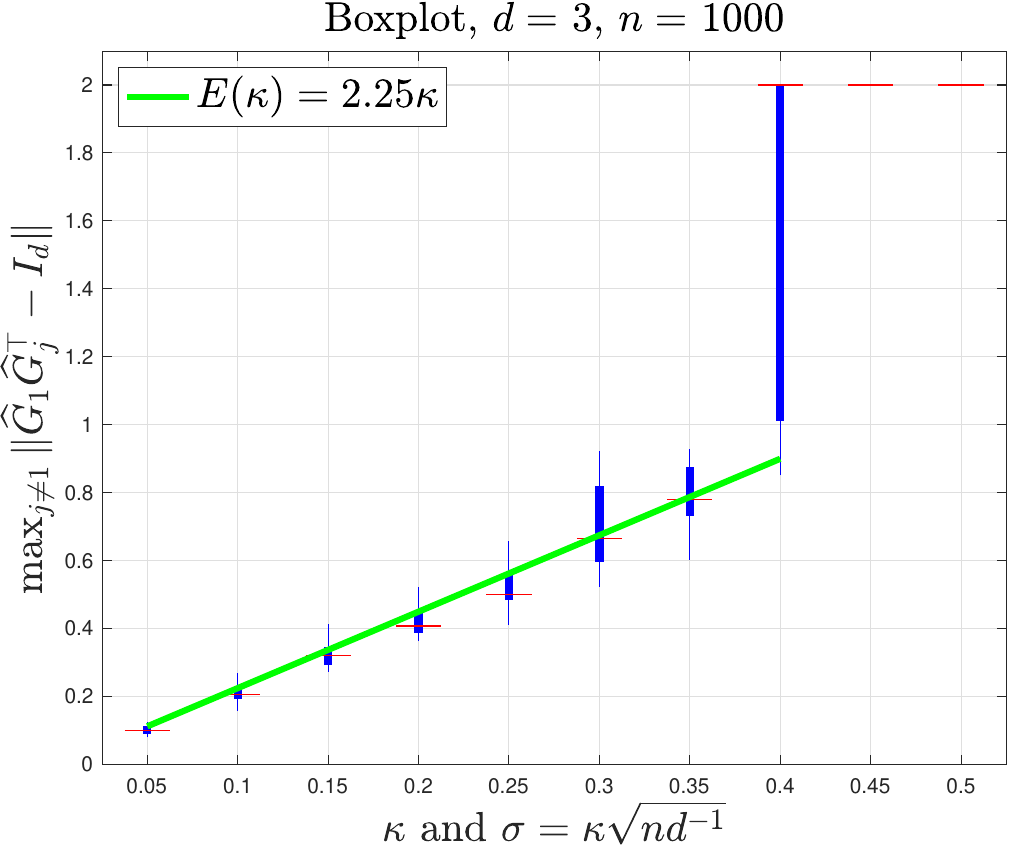}
\end{minipage}
\vskip0.3cm
\vfill
\begin{minipage}{0.48\textwidth}
\includegraphics[width=80mm]{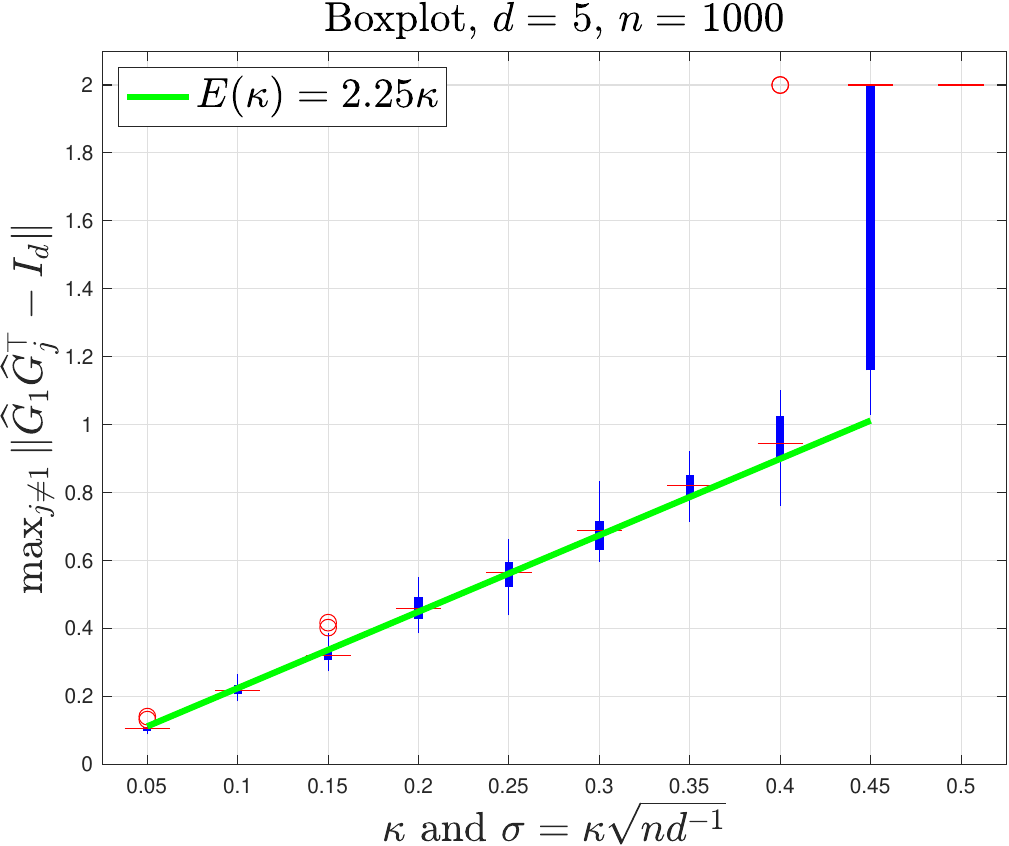}
\end{minipage}
\hfill
\begin{minipage}{0.48\textwidth}
\includegraphics[width=80mm]{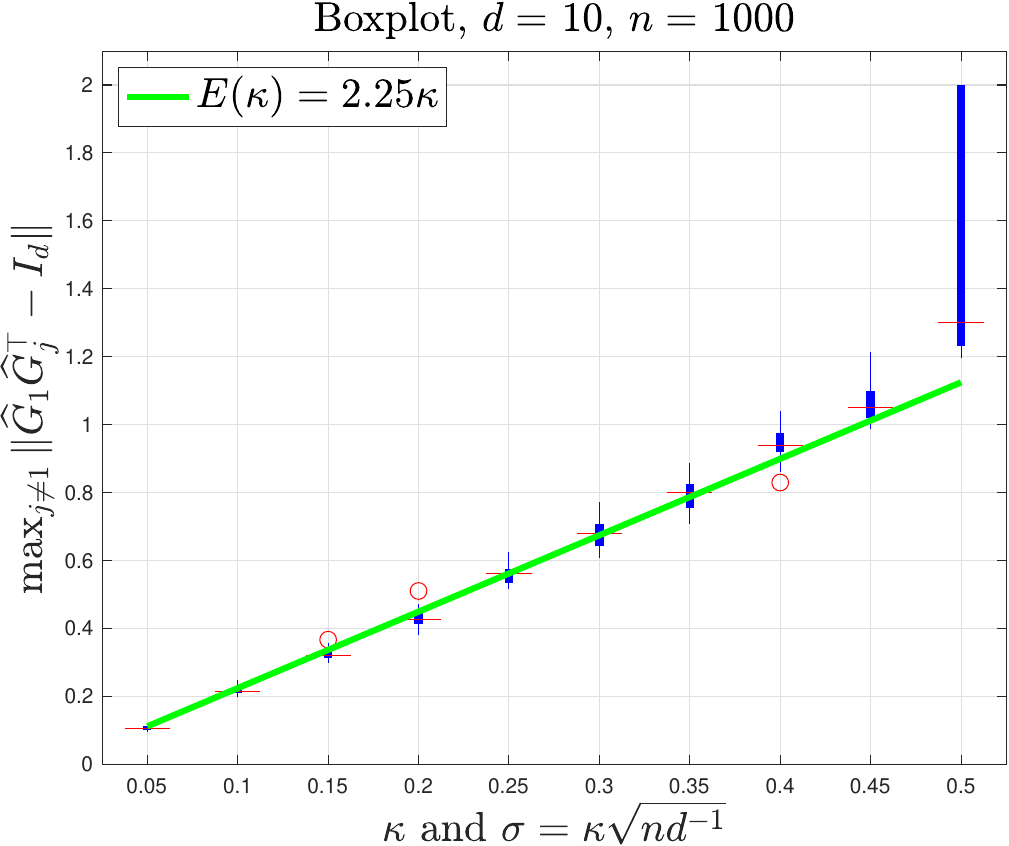}
\end{minipage}
\caption{Blockwise estimation error $E(\kappa)$ v.s. noise level $\kappa$. Each red circle denotes an outlier, the central red mark indicates the median, the bottom/top edges of the blue box indicate the 25th and 75th percentiles respectively. The green line describes the linear relation between the median blockwise estimation error and $\kappa.$}
\label{fig:od_strong}
\end{figure}

\subsection{Permutation group synchronization}
We consider the permutation group synchronization under random corruption. Without loss of generality, we assume $\BG_i = \I_d$ and
$\BG_{ij} = X_{ij}\I_d + (1 - X_{ij})\BP_{ij}$ as introduced in~\eqref{model:pm2}. Then we compute the top $d$ eigenvectors and estimate $\widehat{\BG}_i$ by solving
\[
\widehat{\BG}_i = \argmax_{\BR\in \Pi_d} \left\lag \BR, \PP(\BPhi_1)\PP(\BPhi_i)^{\top} \right\rag.
\]
The goal is to study how the performance of Algorithm 2 depends on the parameters $(n,p)$. Note that our theorem indicates that the algorithm works if $p$ is larger than $\sqrt{n^{-1}\log (nd)}$. Therefore, we introduce the parameter $\kappa$ so that
\[
p = \kappa \sqrt{\frac{\log(nd)}{n}}.
\]
We let $n$ vary from 50 to 1000, and $\kappa$ between $0.1$ and 2.
We use two ways of measuring the recovery performance. 

\vskip0.25cm

{\bf Exact recovery:} We compute the total number of instances in which $\widehat{\BG}_i$ equals $\BG_i$ for all $1\leq i\leq n$. For each pair of $(n,\kappa)$, we run 25 experiments and calculate the proportion of successful instances. Figure~\ref{fig:pm_strong} implies that for $\kappa > 1$, the exact recovery holds with high probability for both $d=5$ and 10. This confirmed the near-optimality of our performance bound in Theorem~\ref{thm:pm}. 

\begin{figure}[h!]
\centering
\begin{minipage}{0.48\textwidth}
\includegraphics[width=80mm]{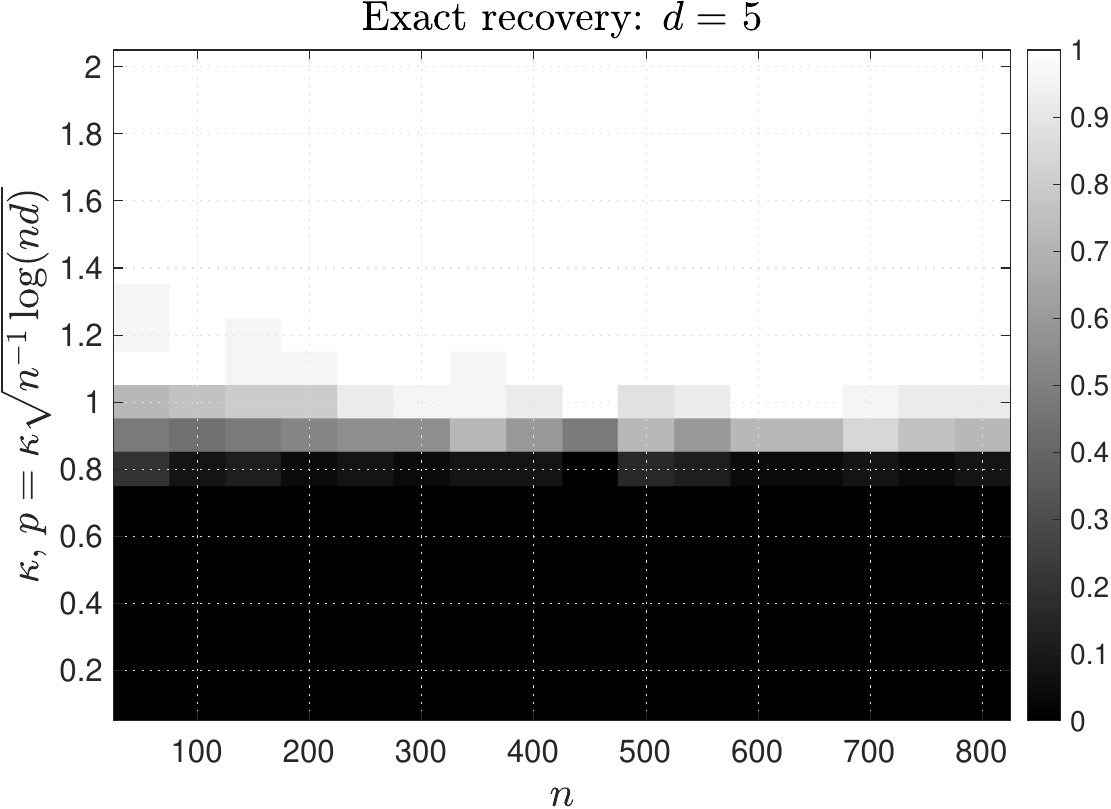}
\end{minipage}
\hfill
\begin{minipage}{0.48\textwidth}
\includegraphics[width=80mm]{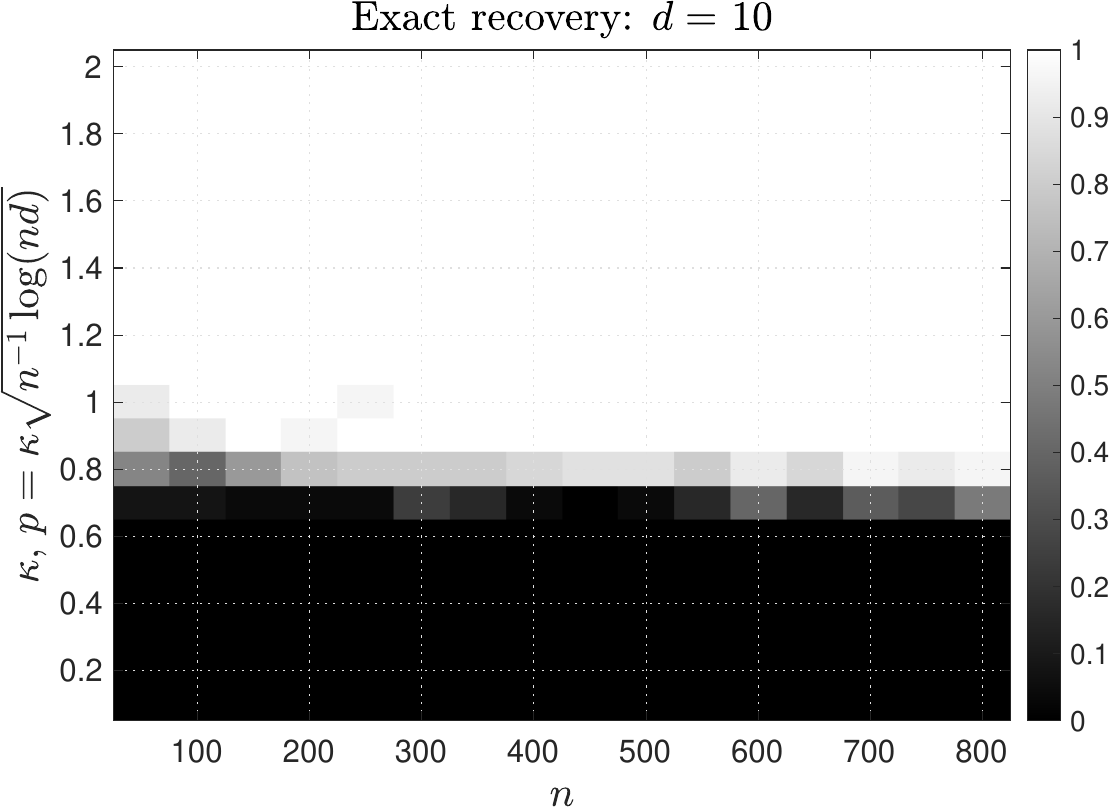}
\end{minipage}
\caption{{Phase transition for the exact recovery of the permutation group synchronization under uniform corruption model. Black region: failure; white region: success. The boundary for the phase transition approximately occurs at $p = \sqrt{n^{-1}\log(nd)}$.}}
\label{fig:pm_strong}
\end{figure}

\vskip0.25cm

{\bf Weak recovery:} Instead of looking at the exact recovery of all the permutation matrices, we compute the~\emph{average alignment} to see if most permutation matrices are recovered even if the assumption of Theorem~\ref{thm:pm} is violated, i.e., adding slightly  more corruption to the observed data.  The average alignment associated with $\{\widehat{\BG}_i\}$ is defined as
\[
\frac{1}{nd}\sum_{i=1}^n \lag \widehat{\BG}_i, \I_d\rag
\]
which is a number between 0 and 1. We simulate 25 instances and the compute the mean of the average alignment. 
The phase transition plot is provided in Figure~\ref{fig:pm_weak}, which shows that if $\kappa> 0.6$, the recovered permutation matrix is highly aligned with the planted signal. We leave the characterization of the critical threshold for the exact/weak recovery to the future work.

\begin{figure}[h!]
\centering
\begin{minipage}{0.48\textwidth}
\includegraphics[width=80mm]{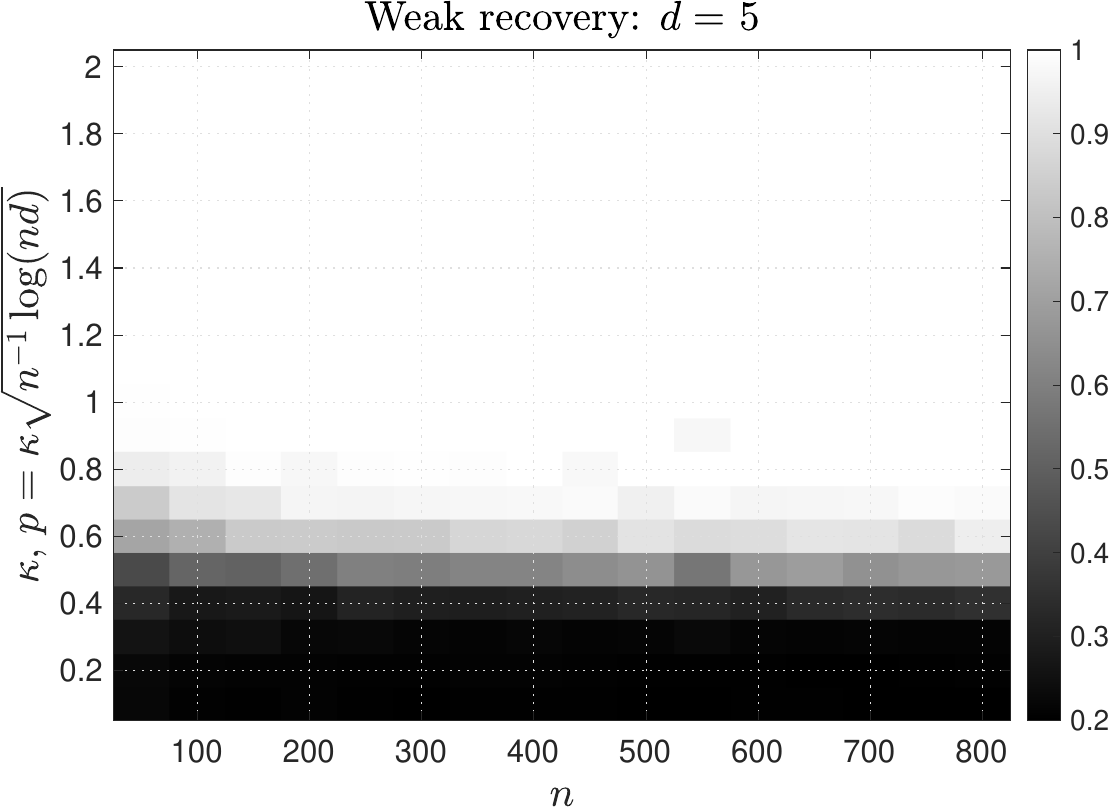}
\end{minipage}
\hfill
\begin{minipage}{0.48\textwidth}
\includegraphics[width=80mm]{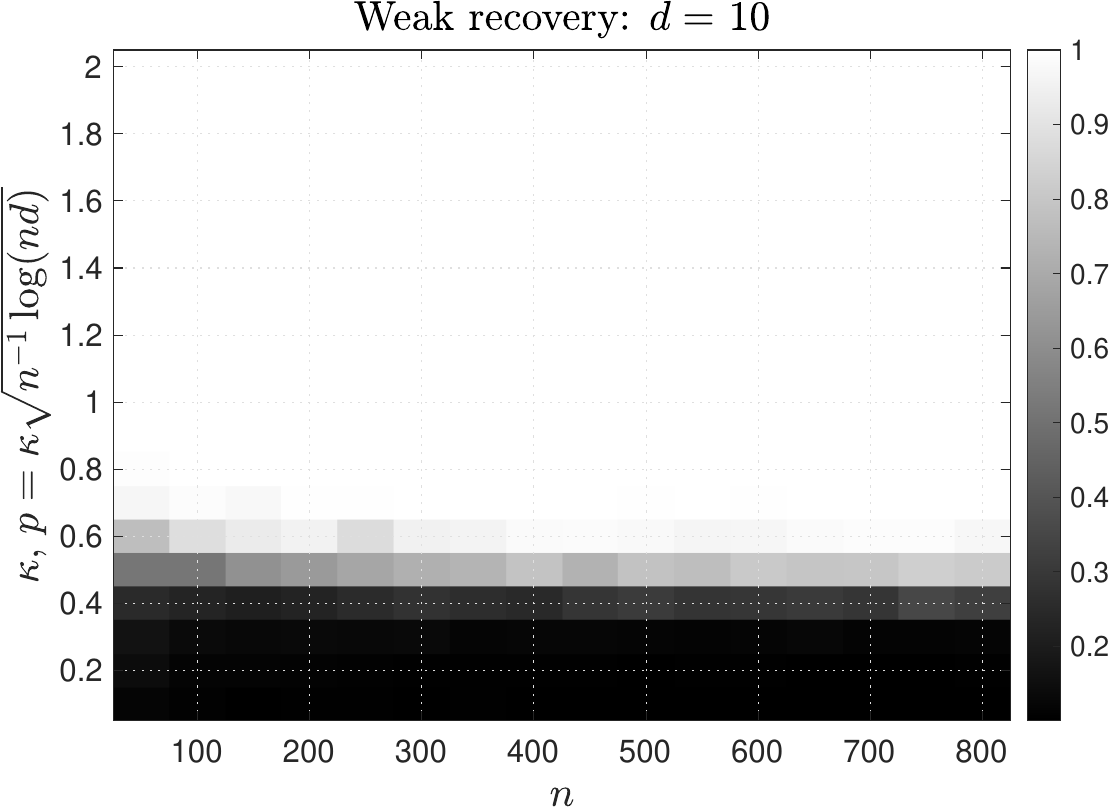}
\end{minipage}
\caption{Phase transition for weak recovery}
\label{fig:pm_weak}
\end{figure}

\section{Proofs}\label{s:proof}
This section is devoted to the proof of Theorem~\ref{thm:od} and~\ref{thm:pm}.

\subsection{Proof of Theorem~\ref{thm:od} and~\ref{thm:pm}}

{As discussed in Section~\ref{ss:loo}, the blockwise estimation of $\BPhi$ is given by~\eqref{eq:keyest}:
\[
\| [\BPhi - \BA\BZ\BQ\BLambda^{-1}]_i \|  \leq \|\BLambda^{-1}\|  \cdot \left(  \| \BZ^{\top}(\BPhi- \BZ\BQ) \| + \|  \BDelta_i^{\top} (\BPhi- \BZ\BQ) \| \right)
\]
where $(\BPhi, \BLambda)$ are the top $d$ eigen-pairs of $\BA$ satisfying $\BA\BPhi = \BPhi\BLambda$ and $\BA = \BZ\BZ^{\top} + \BDelta$ is given in~\eqref{model:ms}. We have shown that the key is to control $\BDelta_i^{\top} (\BPhi - \BZ\BQ)$ by using the leave-one-out technique and it is much easier to bound $\|\BZ^{\top}(\BPhi - \BZ\BQ)\|$ as well as the spectra of $\BLambda.$} 

{Now we will proceed to estimate $\|\BDelta_i^{\top} (\BPhi - \BZ\BQ)\|$ by approximating it with $\|\BDelta_i^{\top} (\BPhi^{(i)} - \BZ\BQ)\|$ where $\BPhi^{(i)}$ is the top $d$ eigenvectors from the auxiliary matrix $\BA^{(i)}$ defined in~\eqref{def:Deltai}. In addition, we also need to bound the distance between $\BPhi$ and $\BPhi^{(i)}$.}
To do that, we formally define an operator which will be frequently used in the discussion. One can view it as a generalization of taking the ``phase" of a matrix which is also known as the matrix sign function~\cite{AFWZ20,G11}.

\begin{definition}
For any matrix $\BX\in\RR^{m\times r}$ with $m\geq r$, we define
\[
\PP(\BX) = \BU\BV^{\top}
\]
where $\BX = \BU\BSigma\BV^{\top}$ is the SVD (singular value decomposition) of $\BX$ with $\BU^{\top}\BU = \BV^{\top}\BV = \I_r$.
In particular, if $\rank(\BX) = r$, i.e., $\BSigma$ is invertible, then
\[
\PP(\BX) = \BX(\BX^{\top}\BX)^{-\frac{1}{2}}.
\]
\end{definition}
{Note that if the input matrix $\BX$ is not full rank, i.e., $\rank(\BX)< r$, then $\PP(\cdot)$ is not unique because the SVD of $\BX$ is not unique. Therefore, it is better to treat $\PP(\BX)$ as a set-valued operator which outputs one representative from the set $\{\BU\BV^{\top}: \BX=\BU\BSigma\BV^{\top}\text{ is the SVD of } \BX\}$.}

The main purpose of introducing $\PP(\cdot)$ is to correctly define the distance among $\BPhi$, $\BPhi^{(i)}$, and $\BZ.$ From now on, we let 
\begin{align}
\begin{split}\label{def:QS}
\BQ & :=  \argmin_{\BR\in \Od} \| \BPhi - \BZ\BR \|_F,  \qquad  \BQ = \PP(\BZ^{\top}\BPhi), \\
{\BQ^{(i)}} &: = \argmin_{\BR\in \Od} \| \BPhi^{(i)} - \BZ\BR \|_F, \quad  \BQ^{(i)} = \PP(\BZ^{\top}\BPhi^{(i)}), \\
{\BS^{(i)} }& :=  \argmin_{\BR\in \Od} \| \BPhi - \BPhi^{(i)}\BR \|_F,  \quad \BS^{(i)} = \PP((\BPhi^{(i)})^{\top}\BPhi),
\end{split}
\end{align}
where the explicit forms of $\BQ$, $\BQ^{(i)}$, and $\BS^{(i)}$ are easy to derive.

Now we decompose $\|\BDelta_i^{\top} (\BPhi - \BZ\BQ) \|$ into three terms and find an upper bound for each of them:
\begin{align}
\| \BDelta_i^{\top} (\BPhi - \BZ\BQ)\|  
& \leq %\| \E\BDelta_i^{\top} (\BPhi - \BZ\BQ)\| \nonumber \\
\|\BDelta_i^{\top} (\BPhi - \BPhi^{(i)}\BS^{(i)} + \BPhi^{(i)}\BS^{(i)} -   \BZ\BQ^{(i)}\BS^{(i)} + \BZ(\BQ^{(i)}\BS^{(i)} - \BQ))\| \nonumber \\
& \leq \|\BDelta_i\|\cdot \|\BPhi - \BPhi^{(i)}\BS^{(i)} \| + \|\BDelta_i^{\top}(  \BPhi^{(i)}-   \BZ\BQ^{(i)})\| + \|\BDelta_i^{\top} \BZ\| \cdot \|\BQ^{(i)} \BS^{(i)} - \BQ\| \nonumber  \\
& =: T_1 + T_2 + T_3. \label{eq:T}
\end{align}
Now it suffices to control each $T_k,1\leq k\leq 3.$

{One may wonder what the main difference of the blockwise analysis of the eigenvectors from the aforementioned works~\cite{AFWZ20,ZB18} is.}  The difference comes from the appearance of $T_3$, which does \emph{not} show up for $d = 1$. Therefore, the estimation of $\|\BQ - \BQ^{(i)}\BS^{(i)}\|$ requires additional treatments for $d\geq 2.$

{Now we present the final estimation of $\|\BPhi_i - [\BA\BZ\BQ]_i\BLambda^{-1}\|$ in~\eqref{eq:keyest}, which is used to establish Theorem~\ref{thm:od} and~\ref{thm:pm}. Note that the proof of Theorem~\ref{thm:keys} naturally includes the estimation of~\eqref{eq:T}. We leave the estimation of~\eqref{eq:T} in Section~\ref{s:proofod} and~\ref{s:proofpm}. } 

\begin{theorem}[{Key theorem}]\label{thm:keys}
Under the assumption of Theorem~\ref{thm:od} and~\ref{thm:pm}, i.e.,
\begin{itemize}
\item For~\eqref{model:od}, 
\begin{equation}\label{def:etaod}
\eta := \sigma n^{-\frac{1}{2}}(\sqrt{d} + \sqrt{\log n}) , \quad \sigma< C_0^{-1} \sqrt{n}(\sqrt{d} + \sqrt{\log n})^{-1}.
\end{equation}

\item For~\eqref{model:pm}, 
\begin{equation}\label{def:etapm}
\eta := p^{-1}n^{-\frac{1}{2}}\sqrt{\log(nd)}, \quad p > C_0\sqrt{n^{-1}\log(nd)}.
\end{equation}
\end{itemize}
Here $C_0 > 0$ is an absolute large constant.
Then with at least $1 - O(n^{-1}d^{-1})$, it holds 
\[
\|\BPhi_i - [\BA\BZ\BQ]_i\BLambda^{-1}\| = \|\BPhi_i - (n\I_d + \BDelta_i^{\top}\BZ)\BQ\BLambda^{-1}\| \lesssim \eta \max_{1\leq i\leq n}\|\BPhi_i\|,
\]
uniformly for all $1\leq i\leq n$. Moreover, we have
\[
| \sigma_j(\BPhi_i) - 1| \lesssim \eta,\qquad \forall 1\leq i\leq n,~1\leq j\leq d.
\]

\end{theorem}

By using the key supporting result Theorem~\ref{thm:keys} above as well as Lemma~\ref{lem:lip} and Claim~\ref{lem:deltaz} below, we provide the proof of Theorem~\ref{thm:od} and~\ref{thm:pm}.
\begin{lemma}\label{lem:lip}
For two $d\times d$ invertible matrices $\BX$ and $\BY$, it holds that
\begin{align*}
\|\PP(\BX)  - \PP(\BY)\| & \leq \frac{2 \|\BX - \BY\|}{\sigma_{\min}(\BX)+ \sigma_{\min}(\BY)},
\end{align*}
where $\sigma_{\min}(\cdot)$ denotes the smallest singular value of a matrix.
\end{lemma}
We will provide a proof of Lemma~\ref{lem:lip} in the appendix which uses the Davis-Kahan theorem. In fact, we found a different proof of the same result in~\cite{L95} after we finished this manuscript. 

\begin{claim}\label{lem:deltaz}
With probability at least $1 - O(n^{-1})$, we have the following results under the assumption of Theorem~\ref{thm:od} and~\ref{thm:pm}.
\begin{itemize}
\item For~\eqref{model:od}, we have
\[
\|\BDelta_i^{\top}\BZ\|\lesssim \sigma\sqrt{n}(\sqrt{d} + \sqrt{\log n}) = n\eta.
\]

\item For~\eqref{model:pm}, we have
\[
\|\BDelta_i^{\top}\BZ\| \lesssim p^{-1}\sqrt{n\log (nd)} = n\eta.
\]
\item It holds that $\|\BDelta\| \leq n/2$ and $\BLambda \succeq \frac{n}{2}\I_n$ for both~\eqref{model:od} and~\eqref{model:pm}.
\end{itemize}
Here $\eta$ is defined in~\eqref{def:etaod} and~\eqref{def:etapm} for both cases respectively.
\end{claim}
The claim will be confirmed in Section~\ref{s:proofod} and~\ref{s:proofpm} for two scenarios respectively. {This claim ensures that $\BPhi_i$ is well approximated by $n\BQ\BLambda^{-1}$ since
\[
\|\BPhi_i - n\BQ\BLambda^{-1}\| \leq \|\BPhi_i - [\BA\BZ\BQ]_i\BLambda^{-1}\| + \| \BDelta_i^{\top}\BZ\BQ\BLambda^{-1} \|
\]
where the two terms are bounded by Theorem~\ref{thm:keys} and Claim~\ref{lem:deltaz} respectively and $[\BA\BZ\BQ]_i = (n\I_d + \BDelta_i^{\top}\BZ)\BQ$ for all $1\leq i\leq n.$ Then applying Lemma~\ref{lem:lip} immediately gives the main results.}

With the results available above, we are ready to provide an upper bound for $\PP(\BPhi_i) - \PP(\BPhi_j)$. In fact, as long as 
$\| \PP(\BPhi_i) - \PP(\BPhi_j)\| \lesssim \eta$, we have
$\| \PP(\BPhi_i)\PP(\BPhi_j)^{\top} -\I_d \| \lesssim \eta, \forall i\neq j.$ This finishes the proof of Theorem~\ref{thm:od}. For Theorem~\ref{thm:pm}, it requires one extra step to show that $\| \PP(\BPhi_i)\PP(\BPhi_j)^{\top} -\I_d \| < 1/2$ holds so that the rounding procedure indeed produces the planted permutation matrices correctly. 

\begin{proof}[{\bf Proof of Theorem~\ref{thm:od} and~\ref{thm:pm}}]
It suffices to estimate $\BPhi_i - \BPhi_j$ and apply Lemma~\ref{lem:lip}. Under Theorem~\ref{thm:keys} and {Claim~\ref{lem:deltaz}}, we have
\begin{align*}
\| \BPhi_i - \BPhi_j\| & \leq \|\BPhi_i - (n\I_d + \BDelta_i^{\top}\BZ)\BQ\BLambda^{-1}\| +  \|\BPhi_j - (n\I_d + \BDelta_j^{\top}\BZ)\BQ\BLambda^{-1}\| \\ &  \qquad + \|  (\BDelta_i - \BDelta_j)^{\top}\BZ\BQ \BLambda^{-1} \|  \\
& \lesssim \eta \max_{1\leq i\leq n}\|\BPhi_i\| + \| (\BDelta_i - \BDelta_j)^{\top} \BZ\| \|\BLambda^{-1}\| \\
& \lesssim \eta \max_{1\leq i\leq n}\|\BPhi_i\| +  n^{-1} \cdot n\eta \lesssim \eta \max_{1\leq i\leq n}\|\BPhi_i\|, \quad i\neq j,
\end{align*}
where the third inequality uses Claim~\ref{lem:deltaz} and $\max_{1\leq i\leq n}\|\BPhi_i\|\geq 1$.
In order to apply Lemma~\ref{lem:lip}, we need to show $\sigma_{\min}(\BPhi_i)$ is away from 0. 

Let $i'$ be the index with the largest $\|\BPhi_{i'}\|$. {Now we will show that $\|\PP(\BPhi_{i'}) - \PP(\BPhi_j)\| \lesssim \eta$ for all $1\leq j\leq n$ and then $\|\PP(\BPhi_{i}) - \PP(\BPhi_j)\| \lesssim \eta$ follows from triangle inequality.} Note that
\[
\|\BPhi - \bone_n\otimes \BPhi_{i'} \| \leq \sqrt{\sum_{j=1}^n \|\BPhi_j - \BPhi_{i'}\|^2} \lesssim \eta\sqrt{n}\|\BPhi_{i'}\|
\]
where $``\otimes"$ denotes the Kronecker product. Note that all the singular values of $\BPhi$ are $\sqrt{n}$ and thus it holds that
\[
|\sqrt{n} - \sqrt{n} \sigma_{\ell}(\BPhi_{i'}) | \lesssim \eta\sqrt{n}\|\BPhi_{i'}\|, \quad \forall 1\leq \ell\leq d. 
\]
This gives $|1 - \sigma_{\min}(\BPhi_{i'})| < \eta\|\BPhi_{i'}\|$ and
$\left|~\|\BPhi_{i'}\| - 1~\right| \lesssim \eta\|\BPhi_{i'}\|$ which imply $\max_{1\leq i\leq n}\| \BPhi_i\| \leq 1 + O(\eta) = O(1)$.
Therefore, 
\[
| 1 - \sigma_{\min}(\BPhi_{i'})|\lesssim\eta, \qquad \| \BPhi_{i'} - \BPhi_j \| \lesssim \eta, \quad \forall 1\leq j\leq n,
\]
which means $\sigma_{\min}(\BPhi_{i'}) > 0.$
Applying Lemma~\ref{lem:lip} gives
\[
\| \PP(\BPhi_{i'}) - \PP(\BPhi_j)\| \leq 2\sigma_{\min}(\BPhi_{i'})^{-1} \| \BPhi_{i'} - \BPhi_j \| \lesssim \eta
\] 
where $|1 - \sigma_{\min}(\BPhi_{i'})| \lesssim \eta.$ By triangle inequality, $\| \PP(\BPhi_{i}) - \PP(\BPhi_j)\| \lesssim\eta$ holds for all pairs of $i$ and $j.$ This finishes the proof of Theorem~\ref{thm:od}.

\vskip0.25cm

For~\eqref{model:pm} with $p > C_0 \sqrt{n^{-1}\log(nd)}$ and a sufficiently large constant $C_0$, it holds that
\[
\left\| \PP(\BPhi_i)\PP(\BPhi_j)^{\top} -\I_d \right\| < \frac{1}{2}.
\]
Then the diagonal entries of $\PP(\BPhi_i)\PP(\BPhi_j)^{\top}$ are its largest $d$ entries {since all the diagonal entries are greater than 1/2 while the off-diagonal entries are smaller than 1/2 in magnitude.} 
Thus the rounding procedure in Algorithm~\ref{algo2} recovers the underlying permutation matrix which is $\I_d$. 
\end{proof}

\subsection{$\Od$ synchronization under Gaussian noise}\label{s:proofod}

This section is devoted to the proof of Theorem~\ref{thm:keys} for~\eqref{model:od}. We first introduce all the necessary ingredients of the proof as follows and leave their proofs later. Then by using these facts, we can prove Theorem~\ref{thm:keys} for~\eqref{model:od}.
The key is to obtain upper bounds for $T_1$, $T_2$, $T_3$, and $\|\BPhi - \BZ\BQ\|$. We provide the bounds for each term below.
\begin{align}
\|\BDelta\| & \lesssim \sigma \sqrt{nd}, \label{eq:Delta} \\
\BLambda & \succeq (n - \|\BDelta\|)\I_d \succeq \frac{n}{2}\I_d, \label{eq:Lambda}\\
\|\BPhi - \BZ\BQ\| & \lesssim \sigma\sqrt{d}, \quad \|\BPhi^{(i)} - \BZ\BQ^{(i)}\| \lesssim \sigma\sqrt{d},  \label{eq:dk1} \\ 
\|\BPhi - \BPhi^{(i)}\BS^{(i)}\| & \lesssim \sigma n^{-\frac{1}{2}}(\sqrt{d} + \sqrt{\log n}) \max_{1\leq i\leq n}\|\BPhi_i\|, \label{eq:dk2} \\
\|\BQ- \BQ^{(i)}\BS\| & \lesssim \sigma n^{-1}(\sqrt{d} + \sqrt{\log n})  \max_{1\leq i\leq n}\|\BPhi_i\|, \label{eq:QS} \\
\|\BDelta_i^{\top}(\BPhi^{(i)} - \BZ\BQ^{(i)}) \| & \lesssim \sigma (\sqrt{d} + \sqrt{\log n})\|\BPhi^{(i)} - \BZ\BQ^{(i)}\| \nonumber \\
& \lesssim \sigma^2 \sqrt{d}(\sqrt{d} + \sqrt{\log n}),  \label{eq:corr1} \\
\|\BDelta_i^{\top}\BZ\| & \lesssim \sigma\sqrt{n}(\sqrt{d} + \sqrt{\log n}).  \label{eq:corr2}
\end{align}
{{\bf \underline{Roadmap}: }The estimation~\eqref{eq:Delta} of $\|\BDelta\|$ simply follows from $\|\BW\|\leq 3\sqrt{nd}$ with high probability for a symmetric Gaussian random matrix $\BW$. In particular, we assume $\|\BDelta\| \leq n/2$ for $\sigma < c_0 \sqrt{nd^{-1}}$ and a small constant $c_0$. The bound~\eqref{eq:Lambda} for the top $d$ eigenvalues of $\BA = \BZ\BZ^{\top} + \BDelta$ follows from Weyl's theorem and~\eqref{eq:Delta} where the top $d$ eigenvalues of $\BZ\BZ^{\top}$ are $n.$
The inequalities~\eqref{eq:dk1} and~\eqref{eq:dk2} follow from Lemma~\ref{lem:dk1} and~\ref{lem:dk2} respectively by using Davis-Kahan theorem (Theorem~\ref{thm:dk}). We prove~\eqref{eq:QS} in Lemma~\ref{lem:QS}, and Lemma~\ref{lem:corr} implies~\eqref{eq:corr1} and~\eqref{eq:corr2}.}

\begin{proof}[{\bf Proof of Theorem~\ref{thm:keys} for~\eqref{model:od}}]

From definition of $T_{\ell}$, $1\leq \ell\leq 3$ and~\eqref{eq:Delta}-\eqref{eq:corr2}, we have
\begin{align*}
T_1 & = \|\BDelta_i\|\cdot \|\BPhi - \BPhi^{(i)}\BS^{(i)} \| \lesssim \sigma^2 \sqrt{d}(\sqrt{d} + \sqrt{\log n}) \max_{1\leq i\leq n}\|\BPhi_i\|, \\
T_2 & =  \|\BDelta_i^{\top}(  \BPhi^{(i)}-   \BZ\BQ^{(i)})\| \lesssim \sigma^2 \sqrt{d}(\sqrt{d} + \sqrt{\log n}), \\
T_3 & = \|\BDelta_i^{\top} \BZ\| \cdot \|\BQ^{(i)} \BS^{(i)} - \BQ\|  \lesssim \sigma^2 n^{-\frac{1}{2}} (\sqrt{d} + \sqrt{\log n})^2 \max_{1\leq i\leq n}\|\BPhi_i\|,
\end{align*}
where $\max_{1\leq i\leq n}\|\BPhi_i\| \geq 1.$
The estimation is bounded by
\[
\| \BDelta_i^{\top} (\BPhi - \BZ\BQ)\|  \leq T_1 + T_2 + T_3 \lesssim \sigma^2 \sqrt{d}(\sqrt{d} + \sqrt{\log n}) \max_{1\leq i\leq n}\|\BPhi_i\|
\]
where $\sigma < c_0\sqrt{nd^{-1}}$ for a small constant $c_0$.
From~\eqref{eq:keyest}, it holds that
\begin{align*}
\|\BPhi_i - (\BZ + \BDelta_i)^{\top}\BZ\BQ\BLambda^{-1}\| 
& \leq \|\BLambda^{-1}\|  \cdot \left(  \| \BZ^{\top}(\BPhi- \BZ\BQ) \| + \|  \BDelta_i^{\top} (\BPhi- \BZ\BQ) \| \right)  \\
& \lesssim n^{-1} (\sqrt{n}\|\BPhi - \BZ\BQ\| + \| \BDelta_i^{\top} (\BPhi - \BZ\BQ)\|  ) \\
& \lesssim n^{-1} \left( \sqrt{n}\cdot \sigma\sqrt{d} + \sigma^2  \sqrt{d}(\sqrt{d} + \sqrt{\log n})   \max_{1\leq i\leq n}\|\BPhi_i\|  \right) \\
& \lesssim \sigma n^{-\frac{1}{2}}(\sqrt{d} + \sqrt{\log n}) \max_{1\leq i\leq n}\|\BPhi_i\|  
\end{align*}
where the first term above follows from $\|\BZ\|= \sqrt{n}$,~\eqref{eq:dk1}, and $\sigma < \sqrt{n/d}.$
\end{proof}

Next we proceed to prove~\eqref{eq:dk1}-\eqref{eq:corr2}.
Our analysis will frequently use the following important fact about Gaussian random matrix.
\begin{theorem}{\bf \cite[Theorem 4.4.5]{V18}}\label{thm:gauss}
For any $\BX\in\RR^{n\times m}$ random matrix whose entries are i.i.d. standard normal random variables. For any $t > 0$, it holds
\[
\|\BX\| \lesssim \sqrt{n} + \sqrt{m} + t
\]
with probability at least $1 - 2\exp(-t^2).$
\end{theorem}
The next two lemmas provide the proof of~\eqref{eq:dk1} and~\eqref{eq:dk2}.

\begin{lemma}[{\bf {Proof of~\eqref{eq:dk1}}}]\label{lem:dk1}
If $\BPhi$ consists of  the top $d$ eigenvectors of $\BA$ with $\BPhi^{\top}\BPhi = n\I_d$, then
\[
\| (\I - n^{-1}\BPhi\BPhi^{\top})\BZ\|  \lesssim \sigma \sqrt{d}.%, \quad \| (\I - n^{-1}\BPhi\BPhi^{\top})\BZ\|_F  \lesssim \sigma d.
\]
The same bound applies to $\BPhi^{(i)}$. Let $\BPhi^{(i)}$ be the eigenvectors associated to the top $d$ eigenvalues of $\BA^{(i)}$ in~\eqref{def:Deltai} and then 
\[
\| (\I - n^{-1}\BPhi^{(i)}(\BPhi^{(i)})^{\top})\BZ\|  \lesssim \sigma \sqrt{d}.%, \quad \| (\I - n^{-1}\BPhi^{(i)}(\BPhi^{(i)})^{\top})\BZ\|_F  \lesssim \sigma d.
\]
Moreover, we have
\[
0\leq n - \sigma_{\min}(\BPhi^{\top}\BZ) \lesssim \sigma \sqrt{nd}, \quad 0\leq n - \sigma_{\min}((\BPhi^{(i)})^{\top}\BZ) \lesssim \sigma \sqrt{nd}.
\]

\end{lemma}

\begin{proof}[{\bf Proof of Lemma~\ref{lem:dk1}}]
We directly apply Davis-Kahan theorem by letting $\BX = \BZ\BZ^{\top}$ and $\BX_{\BE} = \BA = \BZ\BZ^{\top} + \BDelta$ in Theorem~\ref{thm:dk} with $\BDelta = \sigma\BW$. Note that the $d$th largest eigenvalue of $\BA$ is at least $n - \|\BDelta\|$ and the $(d+1)$th eigenvalue of $\BZ\BZ^{\top}$ is 0. Thus set $\delta$ in Theorem~\ref{thm:dk} as $n - \|\BDelta\|$ and it holds
\begin{align*}
\| (\I - n^{-1}\BPhi\BPhi^{\top})\BZ \| & \leq ( n - \|\BDelta\|)^{-1} \| (\BA - \BZ\BZ^{\top})\BZ\| \\
& \lesssim n^{-1}\| (\BA - \BZ\BZ^{\top})\BZ\| \\
& = \sigma n^{-1}\| \BW\BZ\| \\
& \lesssim \sigma n^{-1}\cdot \sqrt{nd}\cdot \sqrt{n} = \sigma \sqrt{d}
\end{align*}
where $\| \BW\BZ\| \leq \|\BW\| \|\BZ\| \lesssim \sqrt{nd} \cdot \sqrt{n} = n\sqrt{d}.$
For $(\BPhi^{(i)}, \BZ)$, simply use $\BDelta^{(i)} = \sigma\BW^{(i)}$ which is defined in~\eqref{def:Deltai} where $\BW^{(i)}$ equals $\BW$ except that its $i$th block row and column are zero. Using Theorem~\ref{thm:dk} again gives
\[
\| (\I - n^{-1}\BPhi^{(i)}(\BPhi^{(i)})^{\top})\BZ \|  \leq \sigma(n - \sigma\|\BW^{(i)}\|)^{-1}\| \BW^{(i)}\BZ\| \lesssim \sigma \sqrt{d}.
\]
Combining Lemma~\ref{lem:mat} together with the results above gives
\[
\|\BPhi - \BZ\BQ\| \leq 2 \| (\I_{nd} - n^{-1}\BPhi\BPhi^{\top})\BZ\|\lesssim \sigma \sqrt{d}, \quad \|\BPhi^{(i)} - \BZ\BQ^{(i)}\| \lesssim \sigma \sqrt{d}
\]
where $\BQ = \PP(\BZ^{\top}\BPhi)\in\RR^{d\times d}$ is an orthogonal matrix. For the singular values of $\BPhi^{\top}\BZ$, we have
\[
 \|\BPhi^{\top}\BZ - n\BQ^{\top}\| = \| (\BPhi - \BZ\BQ)^{\top}\BZ\| \lesssim \sigma\sqrt{nd}.
\]
Then the Weyl's theorem implies that $|n - \sigma_{\min}(\BPhi^{\top}\BZ)| \lesssim \sigma\sqrt{nd}$ for $1\leq \ell\leq d$.
\end{proof}

\begin{lemma}[{\bf {Proof of~\eqref{eq:dk2}}}]\label{lem:dk2}
Let $\BPhi$ and $\BPhi^{(i)}$ be the top $d$ eigenvectors of $\BA$ and $\BA^{(i)}$ in~\eqref{def:Deltai} with $\BPhi^{\top}\BPhi  = (\BPhi^{(i)})^{\top}\BPhi^{(i)}  = n\I_d$ respectively. Then 
\begin{align*}
\| ( \I - n^{-1} \BPhi \BPhi^{\top} )\BPhi^{(i)}\| \lesssim\sigma  n^{-\frac{1}{2}}(\sqrt{d} + \sqrt{\log n}) \cdot \max_{1\leq i\leq n} \|\BPhi_i\|.
\end{align*}
In other words, 
\[
n - \sigma_{\ell}\left( \BPhi^{\top} \BPhi^{(i)} \right) \lesssim  \sigma (\sqrt{d} + \sqrt{\log n})\max_{1\leq i\leq n} \|\BPhi_i\|, \quad 1\leq \ell\leq d
\]
and
\[
\| (\BPhi^{\top} \BPhi^{(i)} (\BPhi^{(i)})^{\top}\BPhi)^{1/2} - n\I_d \| \lesssim \sigma(\sqrt{d} + \sqrt{\log n}) \max_{1\leq i\leq n} \|\BPhi_i\|.
\]
Moreover, it holds
\[
\|\BPhi - \BPhi^{(i)}\BS^{(i)} \|  \leq \sigma n^{-\frac{1}{2}} (\sqrt{d} + \sqrt{\log n})  \max_{1\leq i\leq n}\|\BPhi_i\|
\]
where $\BS^{(i)}$ is defined in~\eqref{def:QS}.
\end{lemma}
\begin{proof}[{\bf Proof of Lemma~\ref{lem:dk2}}]
The $d$th largest eigenvalue of $\BA$ is at least $n - \sigma\|\BW\|$ and the $(d+1)$th largest eigenvalue of $\BA^{(i)}$ is at most $\sigma \|\BW\|.$ Thus we have $\delta\geq n-2\sigma\|\BW\|$ and Theorem~\ref{thm:dk} gives
\begin{align*}
\| ( \I - n^{-1} \BPhi \BPhi^{\top} )\BPhi^{(i)}\| 
& \lesssim n^{-1} \| (\BA - \BA^{(i)})\BPhi^{(i)} \| \\
& = \sigma n^{-1} \| (\BW - \BW^{(i)})\BPhi^{(i)}\|
\end{align*}
where $\BA - \BA^{(i)} = \sigma(\BW - \BW^{(i)})$ holds. Let $\BW_i\in\RR^{nd\times d}$ be the $i$th block column of $\BW$. We have
\[
[ (\BW - \BW^{(i)})\BPhi^{(i)}]_{\ell} = 
\begin{cases}
\BW_{\ell i}\BPhi_i^{(i)}, \quad & \text{ if } \ell\neq i, \\
\BW_i^{\top}\BPhi^{(i)}, \quad & \text{ if } \ell = i.
\end{cases}
\]
Note that $\|\BW_i\|\lesssim\sqrt{nd}$ holds for $1\leq i\leq n$ with high probability.  Each entry of $\BW_i^{\top} \BPhi^{(i)}$ is an independent $\mathcal{N}(0,n)$ random variable. As a result, $\|\BW_i^{\top}\BPhi^{(i)}\| \lesssim \sqrt{n} (\sqrt{d} + \sqrt{\log n}) $ uniformly for all $1\leq i\leq n$ with high probability, following from Theorem~\ref{thm:gauss}.
\begin{align*}
\|(\BW - \BW^{(i)})\BPhi^{(i)} \| & \leq \| \BW_i \BPhi_i^{(i)} \| + \|\BW_i^{\top} \BPhi^{(i)}\| \\
& \lesssim \sqrt{nd} \|  \BPhi_i^{(i)} \| + \sqrt{n}(\sqrt{d} + \sqrt{\log n}) \\
& \lesssim \sqrt{n}(\sqrt{d} + \sqrt{\log n}) (\|  \BPhi_i^{(i)} \| + 1).
\end{align*}

Thus Lemma~\ref{lem:mat} implies that
\begin{align*}
\|\BPhi - \BPhi^{(i)}\BS^{(i)} \| & \leq 2 \| (\I_{nd} - n^{-1}\BPhi\BPhi^{\top})\BPhi^{(i)} \|  \\
&  \lesssim \sigma n^{-1} \cdot  \sqrt{n}(\sqrt{d} + \sqrt{\log n}) (\|  \BPhi_i^{(i)} \|+ 1)\\
& =  \sigma  n^{-\frac{1}{2}}(\sqrt{d} + \sqrt{\log n})  \cdot ( \|  \BPhi_i^{(i)} \|  +1 ).
\end{align*}
Now consider the $i$th block of $\BPhi - \BPhi^{(i)}\BS^{(i)}$ and we have
\[
\| \BPhi^{(i)}_i\| -\| \BPhi_i\| \leq \| \BPhi^{(i)}_i\BS^{(i)} - \BPhi_i\| \lesssim\sigma
n^{-\frac{1}{2}}(\sqrt{d} + \sqrt{\log n})  \cdot ( \|  \BPhi_i^{(i)} \|  +1 ).
\]
This implies that if $\sigma  n^{-\frac{1}{2}}(\sqrt{d} + \sqrt{\log n}) < c_0$ with a sufficiently small constant $c_0$, then
\[
\|\BPhi^{(i)}_i\| \lesssim \max_{1\leq i\leq n}\|\BPhi_i\| 
\] 
where $\max_{1\leq i\leq n}\|\BPhi_i\| \geq 1.$
In other words, it holds
\[
\|\BPhi - \BPhi^{(i)}\BS^{(i)} \|  \lesssim \sigma  n^{-\frac{1}{2}}(\sqrt{d} + \sqrt{\log n})  \max_{1\leq i\leq n}\|\BPhi_i\|
\]
which implies $n - \sigma_{\ell}(\BPhi^{\top}\BPhi^{(i)}) \lesssim \sigma (\sqrt{d} + \sqrt{\log n})  \max_{1\leq i\leq n}\|\BPhi_i\| $ since
\[
\| (\BPhi^{(i)})^{\top}\BPhi - n\BS^{(i)} \| \leq \sqrt{n}\|\BPhi - \BPhi^{(i)}\BS^{(i)} \|  \leq \sigma (\sqrt{d} + \sqrt{\log n})  \max_{1\leq i\leq n}\|\BPhi_i\|.
\]
\end{proof}

\begin{lemma}[{\bf {Proof of~\eqref{eq:QS}}}]\label{lem:QS}
For $\BQ, \BQ^{(i)}$ and $\BS^{(i)}$ defined in~\eqref{def:QS}, we have
\[
\|\BQ - \BQ^{(i)}\BS^{(i)}\| \lesssim \sigma n^{-1}(\sqrt{d} +\sqrt{\log n})\max_{1\leq i\leq n}\|\BPhi_i\|
\]
if~\eqref{eq:dk1} and~\eqref{eq:dk2} hold.
\end{lemma}
\begin{proof}[{\bf Proof of Lemma~\ref{lem:QS}}]
It suffices to estimate
\[
\|\BQ - \BQ^{(i)}\BS^{(i)}\| = \| \PP(\BZ^{\top}\BPhi) - \PP(\BZ^{\top}\BPhi^{(i)}\BS^{(i)}) \|
\]
since $\PP(\BZ^{\top}\BPhi^{(i)}\BS^{(i)}) = \BQ^{(i)}\BS^{(i)}.$
Applying Lemma~\ref{lem:lip} gives
\begin{align*}
\| \PP(\BZ^{\top}\BPhi) - \PP(\BZ^{\top}\BPhi^{(i)}\BS^{(i)}) \| & \leq 2 \sigma^{-1}_{\min}(\BZ^{\top}\BPhi)  \cdot \| \BZ^{\top}(\BPhi - \BPhi^{(i)}\BS^{(i)}) \| \\
& \lesssim n^{-1} \cdot \sqrt{n} \cdot \|\BPhi - \BPhi^{(i)}\BS^{(i)}\| \\
& \lesssim  n^{-1} \cdot \sqrt{n} \cdot  \sigma n^{-\frac{1}{2}}(\sqrt{d} + \sqrt{\log n})   \cdot \max_{1\leq i\leq n} \|\BPhi_i\| \\
& \lesssim \sigma n^{-1}(\sqrt{d} + \sqrt{\log n}) \max_{1\leq i\leq n}\|\BPhi_i\|
\end{align*}
where $\sigma_{\min}^{-1}(  \BZ^{\top}\BPhi) \lesssim n^{-1}$ follows from Lemma~\ref{lem:dk1} and~\ref{lem:dk2}.
\end{proof}

\begin{lemma}[{\bf {Proof of~\eqref{eq:corr1} and~\eqref{eq:corr2}}}]\label{lem:corr}
Suppose a sequence of $n$ matrices $\{\BM_i\}_{i=1}^n$ and $\BM^{\top} = [\BM_1^{\top},\cdots,\BM_n^{\top}]$ which is independent of $\BDelta_i= \sigma\BW_i$. Then 
\[
\| \BDelta_i^{\top}\BM\| \lesssim \sigma \|\BM\| (\sqrt{d} + \sqrt{\log n})
\]
with probability at least $1 - O(n^{-2}).$
In particular, the following inequalities hold with probability at least $1 - O(n^{-1})$
\begin{align*}
\|\BDelta_i^{\top}(\BPhi^{(i)} - \BZ\BQ^{(i)}) \| & \lesssim  (\sqrt{d} + \sqrt{\log n})\|\BPhi^{(i)} - \BZ\BQ^{(i)}\| = \sigma \sqrt{d}(\sqrt{d} + \sqrt{\log n}),  \\
\|\BDelta_i^{\top}\BZ\| & \lesssim \sigma\sqrt{n} (\sqrt{d} + \sqrt{\log n})
\end{align*}
for all $1\leq i\leq n.$
\end{lemma}

\begin{proof}[{\bf Proof of Lemma~\ref{lem:corr}}]
Denote the SVD of $\BM\in\RR^{nd\times d}$ as $\BM = \BU\BSigma\BV^{\top}$, where $\BU\in\RR^{nd\times d}$ with $\BU^{\top}\BU = \I_d$, $\BSigma\in\RR^{d\times d}$, and $\BV\in\RR^{d\times d}$. Note that $\BDelta_i = \sigma\BW_i$ where $\BW_i$ is an $nd\times d$ Gaussian random matrix. 
\begin{align*}
\|\BDelta_i^{\top}\BM\| & = \sigma\|\BW_i^{\top}\BU\BSigma\BV^{\top}\| \\
& \leq \sigma \|\BM\| \|\BW_i^{\top}\BU\|
\end{align*}
where $\|\BSigma\| = \|\BM\|.$ Since $\BU^{\top}\BU = \I_d$, then $\BW_i^{\top}\BU$ is an asymmetric $d\times d$ Gaussian random matrix. Theorem~\ref{thm:gauss} guarantees that $\|\BW_i^{\top}\BU\|$ is bounded by $\|\BW_i^{\top}\BU\|\lesssim \sqrt{d} + \sqrt{\log n}$ with probability at least $1 - O(n^{-2})$. As a result, we have
$\|\BDelta_i^{\top}\BM\| \lesssim \sigma \|\BM\| (\sqrt{d} + \sqrt{\log n}).$

Now by letting $\BM = \BPhi^{(i)}- \BZ\BQ^{(i)}$ or $\BZ$ which is independent of $\BDelta_i$, we have
\[
\|\BDelta_i^{\top}(\BPhi^{(i)} - \BZ\BQ^{(i)})\| \lesssim \sigma \|\BPhi^{(i)} - \BZ\BQ^{(i)}\| (\sqrt{d} + \sqrt{\log n}), \quad \|\BDelta_i^{\top}\BZ\| \lesssim \sigma\sqrt{n}  (\sqrt{d} + \sqrt{\log n})
\]
hold uniformly for all $1\leq i\leq n$ with probability at least $1 - O(n^{-1}).$
\end{proof}

\subsection{Object matching under uniform random corruption}\label{s:proofpm}

Before proceeding to the official proof, we first show that the top $d$ eigenvectors of $\BA_{\BG} = [\BG_{ij}]_{1\leq i,j\leq n}$ are equal to those of $\BA$ in~\eqref{model:ms} and~\eqref{def:Deltapm}. Note that for the spiked matrix model~\eqref{def:Deltapm} for~\eqref{model:pm}, the noise matrix $\BDelta$ is not mean zero, i.e., $\E \BDelta = -p^{-1}(1-p)d^{-1} \I_n\otimes \BJ_d$ is a block-diagonal matrix and $\|\E \BDelta\| \leq p^{-1}(1-p).$
\begin{lemma}\label{lem:AG}
The matrices $ [\BG_{ij}]_{1\leq i,j\leq n}$ and $\BA$ share the same top $d$ eigenvectors.
\end{lemma}
\begin{proof}
Without loss of generality, we assume $\BG_i = \I_d$.
Note that $\BA = \BZ\BZ^{\top} + \BDelta$.
\[
\BDelta_{ij} = 
\begin{cases}
-p^{-1}(1-p)d^{-1}\BJ_d, & \text{ if }i = j, \\
p^{-1}\left( (X_{ij} - p) (\I_d - d^{-1}\BJ_d) + (1-X_{ij})(\BP_{ij} - d^{-1}\BJ_d)\right), & \text{ if } i\neq j.
\end{cases}
\]

Recall that
\[
\BG_{ij} = 
\begin{cases}
X_{ij} \I_d + (1-X_{ij}) \BP_{ij}, & i\neq j, \\
\I_d, & i = j,
\end{cases}
\]
where $X_{ij}\sim$Bernoulli($p$) and $\BP_{ij}$ is a random permutation matrix.

Then for $i\neq j$, 
\[
\E \BG_{ij} = p\I_d + (1 - p)d^{-1}\BJ_d, \quad\E \BP_{ij} = d^{-1}\BJ_d
\]
and
\[
\BG_{ij} - \E \BG_{ij} = (X_{ij} - p)(\I_d - d^{-1}\BJ_{d}) + (1 - X_{ij})(\BP_{ij} - d^{-1}\BJ_d) = p\BDelta_{ij}
\]
which gives
\[
\BG_{ij} =p\BDelta_{ij} + \E\BG_{ij}= p(\I_d + \BDelta_{ij}) + (1-p)d^{-1}\BJ_d = p\BA_{ij} + (1-p)d^{-1}\BJ_d
\]
On the other hand, 
\[
\BA_{ii} = \I_d + \BDelta_{ii} =\I_d - p^{-1}(1-p)d^{-1}\BJ_d \Longleftrightarrow p\BA_{ii} = p\BG_{ii} - (1-p)d^{-1}\BJ_d
\]
As a result, 
\[
p\BA_{ij} = 
\begin{cases}
\BG_{ij} - (1-p)d^{-1}\BJ_d, & i\neq j, \\
p\BG_{ii} - (1-p)d^{-1}\BJ_d, & i= j.
\end{cases}
\]
Thus combining the equations above gives
\[
p\BA = \BA_{G} - (1- p ) d^{-1}\BJ_{nd} - (1-p)\I_{nd} \Longrightarrow 
 \BA_G  = p\BA + (1- p ) d^{-1}\BJ_{nd} + (1-p)\I_{nd}.
\]
Note that $ \BA_{G}$ is a nonnegative matrix with its leading eigenvector $\bone_{nd}$. Also $\BA = \BZ\BZ^{\top} +\BDelta$ has $\bone_{nd}$ as its leading eigenvector if $\|\BDelta\| < n.$
As a result, $\BJ_{nd}$ and $\I_{nd}$ do not change the top $d$ eigenvectors of $\BA$ and $ \BA_G$.
\end{proof}

We follow a similar route of proof as presented in Section~\ref{s:proofod}. Under assumption of Theorem~\ref{thm:keys}, $p > C_0n^{-\frac{1}{2}}\sqrt{\log (nd)}$ for some large constant $C_0$, the following inequalities hold with probability at least $1 - O(n^{-1})$,
\begin{align}
\|\BDelta\| &  \lesssim p^{-1}\sqrt{n\log (nd)} \label{eq:Deltapm} \\
\BLambda & \succeq (n - \|\BDelta\|)\I_d \succeq  \frac{n}{2} \I_d, \label{eq:Lambdapm}\\
\|\BPhi - \BZ\BQ\| & \lesssim  p^{-1}\sqrt{\log(nd)}, \quad \|\BPhi^{(i)} - \BZ\BQ^{(i)}\| \lesssim p^{-1}\sqrt{\log(nd)}, \label{eq:dk1pm} \\ 
\|\BPhi - \BPhi^{(i)}\BS^{(i)}\| & \lesssim p^{-1}n^{-\frac{1}{2}}\sqrt{\log(nd)}\max_{1\leq i\leq n}\|\BPhi_i\|,  \label{eq:dk2pm} \\
\|\BQ- \BQ^{(i)}\BS^{(i)}\| & \lesssim p^{-1}n^{-1}\sqrt{\log(nd)} \max_{1\leq i\leq n}\|\BPhi_i\|, \label{eq:QSpm} \\
\|\BDelta_i^{\top}(\BPhi^{(i)} - \BZ\BQ^{(i)}) \| & \lesssim p^{-1}\sqrt{n\log (nd)}\max_{1\leq i\leq n}\| \BPhi_i \|, \label{eq:corr1pm}  \\
\|\BDelta_i^{\top}\BZ\| & \lesssim  p^{-1}\sqrt{n\log(nd)}.  \label{eq:corr2pm}
\end{align}
{{\bf \underline{Roadmap}:} Here~\eqref{eq:Deltapm} is given in Lemma~\ref{lem:Deltapm} and~\eqref{eq:Lambdapm} directly follows from Weyl's inequality and $\|\BDelta\|\leq n/2$ for $p > C_0n^{-\frac{1}{2}} \sqrt{\log(nd)}$; Lemma~\ref{lem:dk1pm} and~\ref{lem:dk2pm} give~\eqref{eq:dk1pm} and~\eqref{eq:dk2pm} respectively; The estimation of~\eqref{eq:QSpm} is provided in Lemma~\ref{lem:QSpm}; and Corollary~\ref{cor:corrpm_full} implies both~\eqref{eq:corr1pm} and~\eqref{eq:corr2pm}.}

\begin{proof}[{\bf Proof of Theorem~\ref{thm:keys} for~\eqref{model:pm}}]
The proof is highly similar to that of Theorem~\ref{thm:keys} for~\eqref{model:od}. It suffices to estimate $T_{\ell},1\leq \ell\leq 3$ by applying~\eqref{eq:Deltapm}-\eqref{eq:corr2pm} and then plug into~\eqref{eq:keyest} and~\eqref{eq:T}.
\begin{align*}
T_1 & = \|\BDelta_i\|\cdot \|\BPhi - \BPhi^{(i)}\BS^{(i)} \| \\
&  \lesssim p^{-1}\sqrt{n\log (nd)}\cdot p^{-1}n^{-\frac{1}{2}}\sqrt{\log(nd)}\max_{1\leq i\leq n}\|\BPhi_i\|  \\
&  \leq p^{-1}\sqrt{n\log(nd)}\max_{1\leq i\leq n}\|\BPhi_i\|,   \\
T_2 &=  \|\BDelta_i^{\top}(  \BPhi^{(i)}-   \BZ\BQ^{(i)})\| \lesssim p^{-1}\sqrt{n\log(nd)}\max_{1\leq i\leq n}\|\BPhi_i\|, \\
T_3 & = \|\BDelta_i^{\top} \BZ\| \cdot \|\BQ^{(i)} \BS^{(i)} - \BQ\|  \\
& \lesssim p^{-1}\sqrt{n\log(nd)} \cdot (np)^{-1}\sqrt{\log(nd)} \max_{1\leq i\leq n}\|\BPhi_i\| \\
&= p^{-1}\sqrt{\log(nd)} \max_{1\leq i\leq n}\|\BPhi_i\|, 
\end{align*}
where $p^{-1} n^{-\frac{1}{2}}\sqrt{\log (nd))} < 1.$
The estimations in~\eqref{eq:keyest} and~\eqref{eq:T} are bounded by
\[
\| \BDelta_i^{\top} (\BPhi - \BZ\BQ)\|  \leq T_1 + T_2 + T_3 \lesssim p^{-1}\sqrt{n\log(nd)} \max_{1\leq i\leq n}\|\BPhi_i\|
\]
and 
\begin{align*}
\|\BPhi_i - (\BZ + \BDelta_i)^{\top}\BZ\BQ\BLambda^{-1}\| 
& \lesssim n^{-1} (\sqrt{n}\|\BPhi - \BZ\BQ\| + \| \BDelta_i^{\top} (\BPhi - \BZ\BQ)\|  ) \\
& \lesssim n^{-1} ( \sqrt{n}\cdot p^{-1}\sqrt{\log(nd)}+ p^{-1}\sqrt{n\log(nd)} \max_{1\leq i\leq n}\|\BPhi_i\|  ) \\
& \lesssim  p^{-1}n^{-\frac{1}{2}}\sqrt{\log(nd)}\max_{1\leq i\leq n}\|\BPhi_i\|  
\end{align*}
where the first term is bounded by using~\eqref{eq:dk1pm}.
\end{proof}

The estimation of $\|\BDelta\|$ uses the matrix Bernstein inequality, see Theorem~\ref{thm:bern} in the appendix.

\begin{lemma}[{\bf {Proof of~\eqref{eq:Deltapm}}}]\label{lem:Deltapm}
The operator norm of $\BDelta$ is bounded by
\[
\|\BDelta - \E\BDelta\| \lesssim p^{-1}\sqrt{(1-p^2)n\cdot \log (nd)} + p^{-1}\log (nd) 
\]
with probability least $1 - O(n^{-1}d^{-1}).$
In particular, if $p \geq C_0 n^{-\frac{1}{2}}\sqrt{\log (nd)}$ for some large constant $C_0$, 
\[
\|\BDelta \| \leq \|\BDelta - \E\BDelta\|  + \|\E \BDelta\|\lesssim p^{-1}\sqrt{n\log(nd)} \leq \frac{n}{2}.
\] 
where $\|\E\BDelta\| = p^{-1}(1-p).$
\end{lemma}
\begin{proof}
Let $\BZ_{ij}$ be an $nd\times nd$ matrix whose $(i,j)$- and $(j,i)$-block equal $\BDelta_{ij}$ and $\BDelta_{ji}$ respectively and all the other blocks are 0, i.e., 
\[
\BZ_{ij} = \be_i\be_j^{\top}\otimes \BDelta_{ij} + \be_j\be_i^{\top}\otimes \BDelta_{ij}^{\top}
\]
is a symmetric matrix
where $\{\be_i\}_{i=1}^n$ are the canonical basis in $\RR^n.$ Here
\begin{align*}
p\BDelta_{ij} =  (X_{ij} - p)(\I_d - d^{-1}\BJ_d)   + (1 - X_{ij}) (\BP_{ij} - d^{-1}\BJ_d).
\end{align*}
Let's first compute its variance: for $i<j$, we have
\[
\BZ_{ij}\BZ_{ij}^{\top} = \be_i\be_i^{\top}\otimes \BDelta_{ij}\BDelta_{ij}^{\top} + \be_j\be_j^{\top}\otimes \BDelta_{ij}^{\top}\BDelta_{ij}.
\]
By using the independence between $X_{ij}$ and $\BP_{ij}$, the expectations of $\BDelta_{ij}\BDelta_{ij}^{\top}$ and $\BDelta_{ij}^{\top}\BDelta_{ij}$ are
\begin{align*}
\E\BDelta_{ij}\BDelta_{ij}^{\top} & = p^{-2}  \left(\E (X_{ij} - p)^2 (\I_d - d^{-1}\BJ_d) + \E (1-X_{ij})^2(\BP_{ij} - d^{-1}\BJ_d) (\BP_{ij} - d^{-1}\BJ_d)^{\top} \right)  \\
& = p^{-2}(1-p^2)(\I_d - d^{-1}\BJ_d), \\
\E\BDelta_{ij}^{\top} \BDelta_{ij} & =  p^{-2}(1-p^2)(\I_d - d^{-1}\BJ_d)
\end{align*}
where $\E \left[(\BP_{ij} - d^{-1}\BJ_d) (\BP_{ij} - d^{-1}\BJ_d)^{\top} \right] =\E \left[(\BP_{ij} - d^{-1}\BJ_d)^{\top}  (\BP_{ij} - d^{-1}\BJ_d)\right]= \I_d - d^{-1}\BJ_d.$
As a result, it holds that
\begin{align*}
\sum_{i<j} \E \BZ_{ij}\BZ_{ij}^{\top} 
& = p^{-2}(1-p^2) \sum_{i<j} (\be_i\be_i^{\top}+\be_j\be_j^{\top})\otimes (\I_d - d^{-1}\BJ_d) \\
& =  p^{-2}(1-p^2) \sum_{i\neq j}\be_i\be_i^{\top}\otimes (\I_d - d^{-1}\BJ_d)  \\
& = p^{-2}(1-p^2) (n-1)\I_n\otimes (\I_d - d^{-1}\BJ_d)
\end{align*}
which implies $\|\sum_{i<j} \E \BZ_{ij}\BZ_{ij}^{\top} \| \leq p^{-2}(1-p^2)n$.

Note that 
\[
\|\BZ_{ij}\|  \leq p^{-1}\| (X_{ij} - p)\I_d + (1 - X_{ij}) \BP_{ij}  \| \leq 2p^{-1}. 
\]
Applying Bernstein's inequality results in
\[
\|\BDelta - \E\BDelta\| \lesssim p^{-1}\sqrt{(1-p^2)n\cdot \log (nd)} + p^{-1}\log (nd).
\]
with probability at least $1 - O(n^{-1}d^{-1}).$
\end{proof}

\begin{lemma}[{\bf {Proof of~\eqref{eq:dk1pm}}}]\label{lem:dk1pm}
If $\BPhi$ consists of  the top $d$ eigenvectors of $\BA$ with $\BPhi^{\top}\BPhi = n\I_d$, then
\[
\| (\I - n^{-1}\BPhi\BPhi^{\top})\BZ\|  \lesssim p^{-1}\sqrt{\log (nd)}
\]
holds under~\eqref{eq:Deltapm}.
The same bound applies to $\BPhi^{(i)}$. Let $\BPhi^{(i)}$ be the eigenvectors associated to the top $d$ eigenvalues of $\BA^{(i)}$, and then 
\[
\| (\I - n^{-1}\BPhi^{(i)}(\BPhi^{(i)})^{\top})\BZ\|  \lesssim p^{-1}\sqrt{\log (nd)}.
\]
Moreover, we have
\[
n - \sigma_{\min}(\BPhi^{\top}\BZ) \lesssim  p^{-1}\sqrt{n\log (nd)}, \quad n - \sigma_{\min}((\BPhi^{(i)})^{\top}\BZ) \lesssim p^{-1}\sqrt{n\log (nd)}
\]
and
\[
\|\BPhi - \BZ\BQ\| \lesssim p^{-1}\sqrt{\log (nd)}, \quad \|\BPhi^{(i)} - \BZ\BQ^{(i)}\| \lesssim p^{-1}\sqrt{\log (nd)}.
\]
\end{lemma}
\begin{proof}[{\bf Proof of Lemma~\ref{lem:dk1pm}}]
We directly apply Davis-Kahan theorem by letting $\BX = \BZ\BZ^{\top}$ and $\BX_{\BE} = \BA$ in Theorem~\ref{thm:dk}. First we specify the spectral gap:
\[
\delta = |\lambda_d(\BA) - \lambda_{d+1}(\BZ\BZ^{\top})| = |\lambda_d(\BA)|\geq n - \|\BDelta\| \geq n/2
\]
where $\lambda_{d+1}(\BX) = 0.$
\begin{align*}
\| (\I - n^{-1}\BPhi\BPhi^{-1})\BZ \| & \leq ( np - \|\BDelta\|)^{-1} \| (\BA - \BZ\BZ^{\top})\BZ\| \\
& \lesssim n^{-1}\| (\BA - \BZ\BZ^{\top})\BZ\| \\
& \lesssim n^{-1}\| \BDelta\| \cdot \sqrt{n} \\ 
& \lesssim p^{-1}\sqrt{\log (nd)}
\end{align*}
where $\|\BDelta\| \leq p^{-1}\sqrt{n\log(nd)}$ and
$\| \BDelta\BZ\| \leq \|\BDelta\| \|\BZ\| = \sqrt{n}\|\BDelta\|.$
For $(\BPhi^{(i)}, \BZ)$, simply using Theorem~\ref{thm:dk} again with $\BA^{(i)}= \BZ\BZ^{\top} + \BDelta^{(i)}$ leads to
\[
\| (\I - n^{-1}\BPhi^{(i)}(\BPhi^{(i)})^{-1})\BZ \|  \leq (n - \|\BDelta^{(i)}\|)^{-1}\| \BDelta^{(i)}\BZ\| \lesssim p^{-1}\sqrt{\log (nd)}
\]
where $\|\BDelta^{(i)}\| \leq \|\BDelta\| \leq p^{-1}\sqrt{n\log(nd)} < n/4.$

\vskip0.25cm

Combining Lemma~\ref{lem:mat} with the bounds above gives
\[
\|\BPhi - \BZ\BQ\| \lesssim p^{-1}\sqrt{\log (nd)}, \quad \|\BPhi^{(i)} - \BZ\BQ^{(i)}\| \lesssim p^{-1}\sqrt{\log (nd)}.
\]
This provides a lower bound for the smallest singular value of $\BZ^{\top}\BPhi$ and $\BZ^{\top}\BPhi^{(i)}$ which follows from
\[
\|\BZ^{\top}\BPhi - n\BQ\| = \|\BZ^{\top}(\BPhi -\BZ\BQ)\| \lesssim \sqrt{n} \|\BPhi - \BZ\BQ\| \lesssim p^{-1}\sqrt{n\log(nd)}
\]
and the orthogonality of $\BQ.$
\end{proof}

Now we proceed to prove~\eqref{eq:dk2pm}-\eqref{eq:corr2pm} which rely on the following lemma.
\begin{lemma}\label{lem:corrpm}
Let $\BM\in\RR^{nd\times d}$ be a matrix with its $j$th block $\BM_j$, independent of $\BDelta_i$. For each fixed $1\leq i\leq n$, it holds that
\[
\|\BDelta_i^{\top}\BM\|= \left\|\sum_{j=1}^n \BDelta_{ij} \BM_j\right\| \lesssim p^{-1}\sqrt{n\log (nd)} \max_{1\leq j\leq n}\|\BM_j\|
\]
with probability at least $1 - O(n^{-2}d^{-2}).$ 
\end{lemma}

\begin{proof}[{\bf Proof of Lemma~\ref{lem:corrpm}}]
Recall that
\[
\BDelta_{ij} = 
\begin{cases}
-p^{-1}(1-p)d^{-1}\BJ_d, & \text{ if }i = j, \\
p^{-1}\left( (X_{ij} - p) (\I_d - d^{-1}\BJ_d) + (1-X_{ij})(\BP_{ij} - d^{-1}\BJ_d)\right), & \text{ if } i\neq j,
\end{cases}
\]
with $\E \BDelta_{ij} = 0$ for $i\neq j$ and $\BDelta_{ii} = -p^{-1}(1-p)d^{-1}\BJ_d.$ It holds that
\[
\left\|\BDelta_{i}^{\top}\BM\right\| \leq \left\|\sum_{j\neq i} \BDelta_{ij}\BM_j\right\| + \|\BDelta_{ii}\BM_i\| \leq \left\|\sum_{j\neq i} \BDelta_{ij}\BM_j\right\|  + p^{-1}(1-p)\|\BM_i\|.
\]
Now we apply the Bernstein inequality (Theorem~\ref{thm:bern}) to estimate the first term above which is a sum of mean zero independent random matrices.

We first compute $\E \sum_{j\neq i} (\BDelta_{ij}\BM_j)^{\top} (\BDelta_{ij}\BM_j)$. For each $j$, we have
\begin{align*}
\E (\BDelta_{ij}\BM_j)^{\top} (\BDelta_{ij}\BM_j) & = p^{-2} (1-p^2) \BM_j^{\top}(\I_d - d^{-1}\BJ_d)  \BM_j
\end{align*}
where 
\begin{align*}
\E (\BDelta_{ij}^{\top}\BDelta_{ij} ) & = p^{-2}\left(  \E (X_{ij} - p)^2 (\I_d - d^{-1}\BJ_d) +  \E(1-X_{ij})^2 (\BP_{ij} - d^{-1}\BJ_d)(\BP_{ij}^{\top} - d^{-1}\BJ_d) \right) \\
& = p^{-2} (1-p^2) (\I_d - d^{-1}\BJ_d)
\end{align*}
because the cross terms of $\BDelta_{ij}^{\top}\BDelta_{ij}$ are of mean zero and $\E\BP_{ij} = d^{-1}\BJ_d.$
Therefore,
\begin{align*}
\left\| \sum_{j\neq i}\E (\BDelta_{ij}\BM_j)^{\top} (\BDelta_{ij}\BM_j) \right\| 
& \leq p^{-2}(1 - p^2) \left\| \sum_{j\neq i} \BM_j^{\top}(\I_d - d^{-1}\BJ_d)  \BM_j\right\| \\
& \leq p^{-2}(1 - p^2) \sum_{j\neq i} \|\BM_j\|^2 \\
& \leq  p^{-2}(1 - p^2) n \max_{1\leq j\leq n} \|\BM_j\|^2.
\end{align*}

For $\E \sum_{j\neq i} (\BDelta_{ij}\BM_j)(\BDelta_{ij}\BM_j)^{\top} $, we first compute $\E \BDelta_{ij}\BM_j(\BDelta_{ij}\BM_j)^{\top}$:
\begin{align*}
p^2\E (\BDelta_{ij}\BM_j)(\BDelta_{ij}\BM_j)^{\top} & = \E(X_{ij} - p)^2 (\I_d - d^{-1}\BJ_d)\BM_j\BM_j^{\top}(\I_d - d^{-1}\BJ_d)  \\
& \qquad + \E(1- X_{ij})^2 (\BP_{ij} - d^{-1}\BJ_d)\BM_j\BM_j^{\top}(\BP_{ij}^{\top} - d^{-1}\BJ_d) \\
& \preceq \E (X_{ij} - p)^2 \|\BM_j\|^2\I_d + 4 \E(1- X_{ij})^2  \|\BM_j\|^2\I_d \\
%& = p(1-p)(\I_d - d^{-1}\BJ_d)\BM_j\BM_j^{\top}(\I_d - d^{-1}\BJ_d)  \\
%& \qquad + (1-p) \left(\E\BP_{ij}\BM_j\BM_j^{\top}\BP_{ij}^{\top} - d^{-2}\BJ_d\BM_j\BM_j^{\top} \BJ_d\right) \\
& \preceq 4(1-p^2)\| \BM_j\|^2 \I_d.
\end{align*}
Therefore, 
\[
\left\| \sum_{j\neq i}  \E (\BDelta_{ij}\BM_j)(\BDelta_{ij}\BM_j)^{\top} \right\| \leq 4p^{-2}(1-p^2)  \sum_{j\neq i}\left\| \BM_j\right\|^2 \leq 4 p^{-2}(1-p^2)n\max_{1\leq j\leq n}\|\BM_j\|^2
\]
and the variance of $\sum_{j\neq i} \BDelta_{ij}\BM_j$ is bounded by
\begin{align*}
\sigma^2\left( \sum_{j\neq i} \BDelta_{ij}\BM_j \right) %& = \max\left\{\left\| \sum_{j\neq i}  \E (\BDelta_{ij}\BM_j)(\BDelta_{ij}\BM_j)^{\top} \right\|, \left\| \sum_{j\neq i}  \E (\BDelta_{ij}\BM_j)^{\top}(\BDelta_{ij}\BM_j) \right\|  \right\} \\
& \leq 4p^{-2}(1-p^2)n\max_{1\leq j\leq n}\|\BM_j\|^2.
\end{align*}
Each term $\BDelta_{ij}\BM_j$ is bounded by
\[
\|\BDelta_{ij}\BM_j\| \leq \|\BDelta_{ij}\|\|\BM_j\| \leq 2p^{-1}\|\BM_j\|  \leq 2p^{-1}\max_{1\leq j\leq n}\|\BM_j\|, \quad \forall 1\leq j\leq n
\]
where $\|\BDelta_{ij}\| \leq p^{-1} \| (X_{ij} - p)\I_d + (1-X_{ij})\BP_{ij}\|\leq 2p^{-1}$. 
Now applying Bernstein inequality gives
\begin{align*}
\left\|\sum_{j\neq i} \BDelta_{ij}\BM_j \right\| &  \lesssim p^{-1}\sqrt{(1-p^2)n\log (nd)} \max_{1\leq j\leq n}\|\BM_j\| + p^{-1}\log (nd)\max_{1\leq j\leq n}\|\BM_j\| \\
& \lesssim p^{-1}\sqrt{n\log (nd)} \max_{1\leq j\leq n}\|\BM_j\|.
\end{align*}
\end{proof}

\begin{lemma}[{\bf {Proof of~\eqref{eq:dk2pm}}}]\label{lem:dk2pm}
Let $\BPhi$ and $\BPhi^{(i)}$ be the top $d$ eigenvectors of $\BA$ and $\BA^{(i)}$ with $\BPhi^{\top}\BPhi  = (\BPhi^{(i)})^{\top}\BPhi^{(i)}  = n\I_d$ respectively. Then 
\begin{align*}
\|\BPhi - \BPhi^{(i)}\BS^{(i)}\|  \leq 2\| ( \I - n^{-1} \BPhi \BPhi^{\top} )\BPhi^{(i)}\| \lesssim p^{-1} n^{-\frac{1}{2}}  \sqrt{\log (nd)} \max_{1\leq i\leq n} \|  \BPhi_i \|.
\end{align*}
In particular, we have
\[
0\leq n - \sigma_{\min}\left( \BPhi^{\top} \BPhi^{(i)} \right) \lesssim p^{-1}   \sqrt{\log (nd)} \max_{1\leq i\leq n} \|  \BPhi_i \|
\]
and
\[
\| (\BPhi^{\top} \BPhi^{(i)} (\BPhi^{(i)})^{\top}\BPhi)^{\frac{1}{2}} - n\I_d \| \lesssim p^{-1}   \sqrt{\log (nd)} \max_{1\leq i\leq n} \|  \BPhi_i \|.
\]
\end{lemma}
\begin{proof}[{\bf Proof of Lemma~\ref{lem:dk2pm}}]
The $d$th largest eigenvalue of $\BA$ is at least $n - \|\BDelta\|$ and the $(d+1)$th largest eigenvalue of $\BA^{(i)}$ is at most $\|\BDelta\|.$ Thus we have $\delta\geq n-2\|\BDelta\|$ and
\begin{align*}
\| ( \I - n^{-1} \BPhi \BPhi^{\top} )\BPhi^{(i)}\| \lesssim n^{-1} \| (\BA - \BA^{(i)})\BPhi^{(i)} \| =  n^{-1} \| (\BDelta - \BDelta^{(i)})\BPhi^{(i)}\|
\end{align*}
where $\BA - \BA^{(i)} = \BDelta - \BDelta^{(i)}$ holds.
\[
[ (\BDelta - \BDelta^{(i)})\BPhi^{(i)}]_{\ell} = 
\begin{cases}
\BDelta_{\ell i}\BPhi_i^{(i)}, \quad & \text{ if } \ell\neq i, \\
\sum_{k=1}^n \BDelta_{ik}\BPhi_k^{(i)}, \quad & \text{ if } \ell = i.
\end{cases}
\]
Then it holds with probability at least $1 - O(n^{-1}d^{-2})$ that
\begin{align*}
\|(\BDelta - \BDelta^{(i)})\BPhi^{(i)} \| & \leq \| \BDelta_i \BPhi_i^{(i)} \| + \|\BDelta_i^{\top} \BPhi^{(i)}\| \\
& \lesssim \|\BDelta_i\| \cdot \| \BPhi_i^{(i)} \| + p^{-1}\sqrt{n\log (nd)} \max_{1\leq j\leq n}\|\BPhi^{(i)}_j \| \\
& \lesssim p^{-1}\sqrt{n\log (nd)} \max_{1\leq j\leq n} \|  \BPhi_j^{(i)} \|, \quad \forall 1\leq i\leq n
\end{align*}
where $\|\BDelta_i\| \lesssim p^{-1}\sqrt{n\log(nd)}$ uses Lemma~\ref{lem:Deltapm} and $\| \BDelta_i^{\top}\BPhi^{(i)} \| \leq p^{-1}\sqrt{n\log(nd)}\max_{1\leq j\leq n}\|\BPhi_j^{(i)}\|$ follows from the independence between $\BDelta_i$ and $\BPhi^{(i)}$ and Lemma~\ref{lem:corrpm}.

Thus Lemma~\ref{lem:mat} implies that
\begin{align*}
\|\BPhi - \BPhi^{(i)}\BS^{(i)} \| & \leq 2 \| (\I_{nd} - n^{-1}\BPhi\BPhi^{\top})\BPhi^{(i)} \|  \\
&\lesssim n^{-1} \| (\BDelta - \BDelta^{(i)})\BPhi^{(i)}\| \\
&  \lesssim p^{-1}  n^{-\frac{1}{2}} \sqrt{\log (nd)} \max_{1\leq j\leq n} \|  \BPhi_j^{(i)} \|
\end{align*}
holds uniformly for all $1\leq i\leq n$ with probability at least $1 - O(n^{-1}d^{-2})$.
Next we will show that $\max_{1\leq j\leq n} \|  \BPhi_j^{(i)} \| \lesssim \max_{1\leq i\leq n}\|\BPhi_i\|.$
Let $j'$ be the index such that $\|\BPhi_{j'}^{(i)}\| = \max_{1\leq j\leq n}\|\BPhi_j^{(i)}\|$, then applying triangle inequality gives
\[
\| \BPhi^{(i)}_{j'} \| -\| \BPhi_{j'}\| \leq \|\BPhi_{j'} -  \BPhi^{(i)}_{j'} \BS^{(i)}\| \lesssim
p^{-1}n^{-\frac{1}{2}} \sqrt{\log (nd)} \max_{1\leq j\leq n} \|  \BPhi_j^{(i)} \|.
\]
This implies that if $p > C_0n^{-\frac{1}{2}}\sqrt{\log(nd)}$ with a sufficiently large constant $C_0$, then
\begin{equation}\label{eq:Phij}
\max_{1\leq j\leq n}\|\BPhi^{(i)}_j\| \lesssim \|\BPhi_{j'}\| \leq \max_{1\leq j\leq n}\|\BPhi_j\|, \quad \forall 1\leq i\leq n.
\end{equation}
In other words, it holds
\[
\|\BPhi - \BPhi^{(i)}\BS^{(i)} \|  \lesssim p^{-1} n^{-\frac{1}{2}} \sqrt{\log (nd)} \max_{1\leq i\leq n} \|  \BPhi_i \|.
\]
The lower bound for the smallest singular value of $\BPhi^{\top}\BPhi^{(i)}$ directly follows from
\[
\|(\BPhi^{(i)})^{\top}\BPhi - n\BS^{(i)} \|  \leq \sqrt{n}\|\BPhi - \BPhi^{(i)}\BS^{(i)} \| \leq p^{-1}  \sqrt{\log (nd)} \max_{1\leq i\leq n} \|  \BPhi_i \|
\]
and $(\BPhi^{(i)})^{\top}\BPhi^{(i)} = n\I_d.$
\end{proof}

\begin{lemma}[{\bf {Proof of~\eqref{eq:QSpm}}}]\label{lem:QSpm}
Under~\eqref{eq:dk1pm} and~\eqref{eq:dk2pm}, the three orthogonal matrices $(\BQ,\BQ^{(i)},\BS^{(i)})$ defined in~\eqref{def:QS} satisfy
\[
\|\BQ - \BQ^{(i)}\BS^{(i)}\| \lesssim p^{-1}n^{-1}  \sqrt{\log (nd)} \max_{1\leq j\leq n} \|  \BPhi_j \|, \quad \forall 1\leq i\leq n.
\]
\end{lemma}
\begin{proof}[{\bf Proof of Lemma~\ref{lem:QSpm}}]
The proof is exactly the same as that of Lemma~\ref{lem:QS} except under a different setting. For the completeness of presentation, we still provide the proof here.
An important observation is that
\[
\BQ - \BQ^{(i)}\BS^{(i)} = \PP(\BZ^{\top}\BPhi) - \PP(\BZ^{\top}\BPhi^{(i)}\BS^{(i)})
\]
since $\PP(\BZ^{\top}\BPhi^{(i)}\BS^{(i)}) = \BQ^{(i)}\BS^{(i)}$ and $\PP(\BZ^{\top}\BPhi^{(i)}) = \BQ^{(i)}.$
Note that $n - \sigma_{\min}(\BZ^{\top}\BPhi) \lesssim p^{-1}\sqrt{n\log(nd)}$ which follows from Lemma~\ref{lem:dk1pm}. This means $\sigma_{\min}(\BZ^{\top}\BPhi) \geq n/2$ if $p > C_0 n^{-\frac{1}{2}}\sqrt{\log(nd)}$ for a sufficiently large constant $C_0.$
An upper bound of $\|\BQ - \BQ^{(i)}\BS^{(i)}\|$ can be found by applying Lemma~\ref{lem:lip}:
\begin{align*}
\| \PP(\BZ^{\top}\BPhi) - \PP(\BZ^{\top}\BPhi^{(i)}\BS^{(i)}) \| & \leq 2 \sigma^{-1}_{\min}(\BZ^{\top}\BPhi)  \cdot \| \BZ^{\top}(\BPhi - \BPhi^{(i)}\BS^{(i)}) \| \\
& \lesssim n^{-1} \cdot \sqrt{n} \cdot \|\BPhi - \BPhi^{(i)}\BS^{(i)}\| \\
& \lesssim  n^{-1} \cdot \sqrt{n} \cdot  p^{-1} n^{-\frac{1}{2}} \cdot  \sqrt{\log (nd)} \max_{1\leq j\leq n} \|  \BPhi_j \| \\
& \lesssim p^{-1}n^{-1}  \sqrt{\log (nd)} \max_{1\leq j\leq n} \|  \BPhi_j \|
\end{align*}
where 
$\| \BPhi - \BPhi^{(i)}\BS^{(i)} \|\lesssim  p^{-1} n^{-\frac{1}{2}}  \sqrt{\log (nd)} \max_{1\leq i\leq n} \|  \BPhi_i \|$ is given in Lemma~\ref{lem:dk2pm}.
\end{proof}
{
\begin{corollary}[{\bf Proof of~\eqref{eq:corr1pm} and~\eqref{eq:corr2pm}}]\label{cor:corrpm_full}

With probability at least $1-O(n^{-1}d^{-2})$, 
\[
\|\BDelta_i^{\top}(\BPhi^{(i)} - \BZ\BQ^{(i)})\| \lesssim p^{-1}\sqrt{n\log (nd)} \max_{1\leq i\leq n}\| \BPhi_i \|
\]
and
\[
\|\BDelta_i^{\top}\BZ\| \lesssim p^{-1}\sqrt{n\log (nd)}
\]
 %with probability at least $1 - O(n^{-1}d^{-2})$ 
%If $\BM= \BPhi^{(i)} - \BZ\BQ^{(i)}$, i.e., $\BM_j =  \BPhi^{(i)}_j - \BQ^{(i)}$, then
hold uniformly for all $1\leq i\leq n$.

\end{corollary}}
\begin{proof}

The proof of~\eqref{eq:corr1pm} and~\eqref{eq:corr2pm} directly follows from Lemma~\ref{lem:corrpm}. 
For~\eqref{eq:corr1pm}, we let $\BM = \BPhi^{(i)} - \BZ\BQ^{(i)}$ in Lemma~\ref{lem:corrpm}. Note that $\BPhi^{(i)}$ is the top $d$ eigenvectors of $\BA^{(i)}$ which is independent of $\BDelta_i$. Thus we can apply the concentration bound above and the following holds with probability at least $1 - O(n^{-1})$
\[
\|\BDelta_i^{\top}(\BPhi^{(i)} - \BZ\BQ^{(i)})\| \lesssim p^{-1}\sqrt{n\log (nd)} \max_{1\leq j\leq n} \|\BPhi^{(i)}_j - \BQ^{(i)}\|.
\]
In the proof of Lemma~\ref{lem:dk2pm}, we have~\eqref{eq:Phij}, i.e., $\max_{1\leq j\leq n}\|\BPhi_j ^{(i)}\|\lesssim \max_{1\leq j\leq n}\|\BPhi_j\|$. As a result, we have $\|\BPhi^{(i)}_j - \BQ^{(i)}\| \leq \|\BPhi^{(i)}_j \| + 1\lesssim 2\max_{1\leq j\leq n}\|\BPhi_j\|$ and thus
\[
\|\BDelta_i^{\top}(\BPhi^{(i)} - \BZ\BQ^{(i)})\| \lesssim p^{-1}\sqrt{n\log (nd)} \max_{1\leq j\leq n} \|\BPhi_j\|, \quad 1\leq i\leq n.
\]

It is easier to show~\eqref{eq:corr2pm} holds with probability at least $1 - O(n^{-1})$ by simply choosing $\BM = \BZ$, i.e., $\BM_j = \I_d$, and taking the union bound over $1\leq i\leq n$. 
\end{proof}

\section{Conclusion}

To conclude this work, we discuss a few future directions beyond our current results. 
Our model assumes that the underlying network is complete, i.e., all the pairwise measurements among these group elements are taken. However, the network is usually very sparse in practice, especially in computer vision and imaging sciences. Therefore, for the group synchronization on general networks, it would be very interesting to analyze the spectral methods based on the (normalized) connection Laplacian~\cite{BSS13,SLL20} associated to $\BA_G$ or to study the cycle-edge message passing type algorithm~\cite{LS19}. For the spectral methods based on connection Laplacian, we would encounter new technical difficulties in deriving the blockwise error bound for the bottom eigenvectors of the corresponding (normalized) connection Laplacian. This is because the columns/rows of the connection Laplacian are no longer block-wisely independent, which is crucial in the current theoretical framework. The similar technical issue would also appear when we deal with non-uniform noise scenario.
Another possible direction is extending the leave-one-out technique to the synchronization problem over non-compact groups, for example, the additive group over the real line~\cite{DCT21} and the special Euclidean group~\cite{ARF16,CLS12,RCBL19,LHBC19} under certain statistical models. We leave all these topics to the future work.

\section*{Appendix: important technical ingredients}\label{s:app}
We list all the necessary supporting results in this section.

\begin{theorem}[{\bf Weyl's inequality~\cite{S98}}]\label{thm:w}
For two matrices $\BX$ and $\BY$ of the same size, it holds
\[
|\sigma_{\ell}(\BX) - \sigma_{\ell}(\BY)| \leq \|\BX - \BY\|, \quad \forall \ell
\]
where $\sigma_{\ell}(\cdot)$ denotes the $\ell$th largest singular value of a matrix.
\end{theorem}

\begin{theorem}[\bf Davis-Kahan theorem~\cite{DK70}]\label{thm:dk}
Let $\BX$ and $\BX_E = \BX + \BDelta$ be two symmetric matrices. Suppose $\BPsi_1$ and $\BPsi_{1,\BDelta}$ are the top $d$ eigenvectors of $\BX$ and $\BX_{\BDelta}$ respectively. 
\[
\BX = 
\begin{bmatrix}
\BPsi_1 & \BPsi_2
\end{bmatrix}
\begin{bmatrix}
\BLambda_1 & 0 \\
0 & \BLambda_2
\end{bmatrix}
\begin{bmatrix}
\BPsi_1 & \BPsi_2
\end{bmatrix}^{\top}, ~
\BX_{\BDelta} = 
\begin{bmatrix}
\BPsi_{1,\BDelta} & \BPsi_{2,\BDelta}
\end{bmatrix}
\begin{bmatrix}
\BLambda_{1,\BDelta} & 0 \\
0 & \BLambda_{2,\BDelta}
\end{bmatrix}
\begin{bmatrix}
\BPsi_{1,\BDelta} & \BPsi_{2,\BDelta}
\end{bmatrix}^{\top}
\]
where the columns of $\BPsi_{k}$ and $\BPsi_{k,\BDelta}$ are normalized for $k=1,2$, and $\BLambda_{k}$ and $\BLambda_{k,\BDelta}$ are diagonal matrices with the corresponding eigenvalues. Then it holds that
\[
\| (\I - \BPsi_{1,\BDelta}\BPsi_{1,\BDelta}^{\top})\BPsi_1 \|\leq \frac{ \| \BDelta \BPsi_1 \|}{\delta}
\]
where $\delta$ denotes the spectral gap between $\BLambda_{1,\BDelta}$ and  $\BLambda_2$, i.e., 
$\delta = |\lambda_{\min}(\BLambda_{1,\BDelta}) - \lambda_{\max}(\BLambda_2) |$.
\end{theorem}
Theorem~\ref{thm:dk} is a classical result in matrix perturbation theory.

\begin{lemma}\label{lem:mat}
Suppose $\BX$ and $\BY$ are two tall orthogonal matrices of the same size $n\times r$, i.e., $\BX^{\top}\BX = \BY^{\top}\BY = \I_r$, then
\begin{equation*}
\|\BY - \BX\BR\| \leq 2 \| (\I_n - \BX\BX^{\top})\BY \|, 
\end{equation*}
where 
$\BR = \PP(\BX^{\top}\BY).$

%If $\| (\I_n - \BX\BX^{\top})\BY  \| \leq \eps$, then 
%\[\sqrt{1- \eps} \leq \sigma_{\min}(\BX^{\top}\BY) \leq \sigma_{\max}(\BX^{\top}\BY) \leq 1.\]

\end{lemma}
\begin{proof}
Suppose $\BU\BSigma\BV^{\top}$ is the SVD of $\BX^{\top}\BY.$ Then $\BR = \BU\BV^{\top}\in\RR^{r\times r}$ is orthogonal.
\begin{align*}
\BY - \BX\BR = (\I_n - \BX\BX^{\top})\BY + \BX(\BX^{\top}\BY - \BU\BV^{\top})
%&= \BY - \BX \BX^{\top}\BY(\BY^{\top}\BX\BX^{\top}\BY)^{-\frac{1}{2}} \\
%& = (\I_n - \BX \BX^{\top})\BY + \BX \BX^{\top}\BY(\I_r - (\BY^{\top}\BX\BX^{\top}\BY)^{-\frac{1}{2}} )
\end{align*}
Taking the operator norm yields
\begin{align*}
\|\BY - \BX\BR \| & \leq \|(\I_n - \BX \BX^{\top})\BY\| 
+ \| \BX (\BU\BSigma\BV^{\top} -\BU\BV^{\top}) \|  \\
& \leq \|(\I_n - \BX \BX^{\top})\BY\| + \|\BSigma - \I_r\|
\end{align*}
For $\|\BSigma - \I_r\|$, it suffices to find a lower bound for the smallest singular value of $\BX^{\top}\BY$ since
\[
\|\BSigma - \I_r\| \leq 1 - \sigma_{\min}(\BX^{\top}\BY)
\]
and all the singular values of $\BX^{\top}\BY$ are no larger than 1. Note that
\[
 1 - \sigma_{\min}^2(\BX^{\top}\BY) = \| \I_r - \BY^{\top}\BX\BX^{\top}\BY\| \leq\|\BY^{\top}(\I_n - \BX\BX^{\top})\BY\|\leq \| (\I_n - \BX\BX^{\top})\BY\|
\]
which implies 
\[
1 - \sigma_{\min}(\BX^{\top}\BY)\leq 1-\sigma_{\min}^2(\BX^{\top}\BY) \leq \| (\I_n - \BX\BX^{\top})\BY\| 
\]
where $0\leq \sigma_{\min}(\BX^{\top}\BY)\leq 1.$
As a result, $\|\BY - \BX\BR\| \leq 2\|(\I_n - \BX\BX^{\top})\BY\|$ holds.
\end{proof}

\begin{proof}[{\bf Proof of Lemma~\ref{lem:lip}}]
Let $\BX = \BU_{\BX}\BSigma_{\BX}\BV_{\BX}^{\top}\in\RR^{d\times d}$ and $\BY = \BU_{\BY}\BSigma_{\BY}\BV_{\BY}^{\top}\in\RR^{d\times d}$ be the SVD of $\BX$ and $\BY$ respectively. Here $\BSigma_{\BX}$ is a $d\times d$ PSD (positive semidefinite) matrix which consists of the singular values of $\BX$ and the same applies to $\BSigma_{\BY}.$
Note that the goal here is to estimate the difference between $\BU_{\BX}\BV_{\BX}^{\top} - \BU_{\BY}\BV_{\BY}^{\top}$ and it suffices to bound the difference between $\BU_{\BX}$ and $\BU_{\BY}$, and that between $\BV_{\BX}$ and $\BV_{\BY}$. We will apply the Davis-Kahan theorem to obtain such an upper bound by considering the augmented matrix.
Define the augmented matrix of $\BX$ and $\BY:$
\[
\widetilde{\BX} =
\begin{bmatrix}
0 & \BX \\
\BX^{\top} & 0
\end{bmatrix}, \quad 
\widetilde{\BY} =
\begin{bmatrix}
0 & \BY \\
\BY^{\top} & 0
\end{bmatrix}
\]
and 
\[
\BM_{\BX} :=\frac{1}{\sqrt{2}}
\begin{bmatrix}
\BU_X \\
\BV_X
\end{bmatrix}, \quad
\BM_{\BY} := \frac{1}{\sqrt{2}}
\begin{bmatrix}
\BU_Y \\
\BV_Y
\end{bmatrix}.
\]
It is well-known in linear algebra that $\BM_{\BX}$ and $\BM_{\BY}$ are the normalized eigenvectors of $\widetilde{\BX}$ and $\widetilde{\BY}$ and the corresponding eigenvalues are the singular values of $\BX$ and $\BY$ respectively: $\widetilde{\BX}\BM_{\BX} = \BM_{\BX}\BSigma_{\BX}.$ The other bottom $d$ nonzero eigenvalues of $\widetilde{\BX}$ and $\widetilde{\BY}$ are given by the negative singular values of $\BX$ and $\BY$ respectively. 
Applying the Davis-Kahan theorem (Theorem~\ref{thm:dk}) gives
\begin{align*}
\| (\I_{2d} - \BM_{\BX}\BM_{\BX}^{\top})\BM_{\BY}\|
& \leq \frac{1}{\lambda_{d}(\widetilde{\BX})- \lambda_{d+1}(\widetilde{\BY})} \cdot\| (\widetilde{\BX} - \widetilde{\BY})\BM_{\BY} \| \\
& \leq  \frac{1}{\sigma_{\min}(\BX)+ \sigma_{\min}(\BY)} \cdot \|\BX - \BY\|\|\BM_{\BY}\| \\
& \leq  \frac{ \|\BX - \BY\|}{\sigma_{\min}(\BX)+ \sigma_{\min}(\BY)}
\end{align*}
where the spectral gap $\lambda_{d}(\widetilde{\BX})- \lambda_{d+1}(\widetilde{\BY})$ equals $\sigma_{\min}(\BX)+\sigma_{\min}(\BY)$ since all the eigenvalues of $\widetilde{\BX}$ and $\widetilde{\BY}$ are $\{\pm \sigma_{\ell}(\BX)\}_{\ell=1}^{d}$ and $\{\pm \sigma_{\ell}(\BY)\}_{\ell=1}^{d}$ respectively. Using the definition of $\BM_X$ and $\BM_Y$, we have
\begin{align*}
(\I_{2d} - \BM_{\BX}\BM_{\BX}^{\top})\BM_{\BY} & = \frac{1}{2\sqrt{2}}
\begin{bmatrix}
\I_d & -\BU_X\BV_X^{\top} \\
-\BV_X\BU_X^{\top} & \I_d
\end{bmatrix}  
\begin{bmatrix}
\BU_Y \\
\BV_Y
\end{bmatrix} \\
& = \frac{1}{2\sqrt{2}}
\begin{bmatrix}
\BU_Y  -\BU_X\BV_X^{\top}\BV_Y \\
\BV_Y -\BV_X\BU_X^{\top} \BU_Y
\end{bmatrix}  \\
& = \frac{1}{2\sqrt{2}} 
\begin{bmatrix}
-\I_d & 0\\
0 & \BV_X\BU_X^{\top}
\end{bmatrix}
\begin{bmatrix}
\BU_X\BV_X^{\top} - \BU_Y\BV_Y^{\top} \\
\BU_X\BV_X^{\top} - \BU_Y\BV_Y^{\top}
\end{bmatrix}
\BV_Y
\end{align*}
Note that $\BU_X,\BV_X,\BU_Y,$ and $\BV_Y$ are all orthogonal. Thus
\[
\|(\I_{2d} - \BM_{\BX}\BM_{\BX}^{\top})\BM_{\BY}\| = \frac{1}{2}\| \BU_X\BV_X^{\top} - \BU_Y\BV_Y^{\top}\| = \frac{1}{2}\| \PP(\BX) - \PP(\BY) \|.
\]

As a result, we have
\[
\| \PP(\BX) - \PP(\BY) \| = 2\|(\I_{2d} - \BM_{\BX}\BM_{\BX}^{\top})\BM_{\BY}\|   \leq \frac{2\|\BX-\BY\|}{\sigma_{\min}(\BX)+\sigma_{\min}(\BY)}.
\]

\end{proof}

\begin{theorem}[Matrix Bernstein~\cite{T12}]\label{thm:bern}
Consider a finite sequence $\{\BZ_k\}$ of independent random matrices. Assume that each random matrix satisfies 
\[
\E(\BZ_k) = 0, \quad \| \BZ_k\| \leq R
\]
Then for all $t \geq 0$,
\[
\Pr\left( \left\| \sum_k \BZ_k \right\| \geq t\right) \leq (d_1+d_2)\cdot \exp\left( -\frac{t^2/2}{\sigma^2 + Rt/3} \right)
\]
where
\[
\sigma^2 =\max\left\{ \left\|  \sum_{k} \E \BZ_k^{\top}\BZ_k \right\|, \left\|  \sum_{k} \E \BZ_k\BZ_k^{\top} \right\|\right\}.
\]

\end{theorem}

It is easy to see that
\[
\left\| \sum_k \BZ_k \right\| \leq \sqrt{2\gamma\sigma^2 \log (d_1 + d_2)} + \frac{2\gamma R\log(d_1+d_2)}{3} 
\]
with probability at least $1 - n^{-\gamma + 1}.$

%\bibliography{SpectraBlock.bib}

\begin{thebibliography}{10}

\bibitem{ABBS14}
E.~Abbe, A.~S. Bandeira, A.~Bracher, and A.~Singer.
\newblock Decoding binary node labels from censored edge measurements: Phase
  transition and efficient recovery.
\newblock {\em IEEE Transactions on Network Science and Engineering},
  1(1):10--22, 2014.

\bibitem{AFWZ20}
E.~Abbe, J.~Fan, K.~Wang, and Y.~Zhong.
\newblock Entrywise eigenvector analysis of random matrices with low expected
  rank.
\newblock {\em Annals of Statistics}, 48(3):1452, 2020.

\bibitem{AMS09}
P.-A. Absil, R.~Mahony, and R.~Sepulchre.
\newblock {\em Optimization Algorithms on Matrix Manifolds}.
\newblock Princeton University Press, 2009.

\bibitem{AKKSB12}
M.~Arie-Nachimson, S.~Z. Kovalsky, I.~Kemelmacher-Shlizerman, A.~Singer, and
  R.~Basri.
\newblock Global motion estimation from point matches.
\newblock In {\em 2012 Second International Conference on 3D Imaging, Modeling,
  Processing, Visualization \& Transmission}, pages 81--88. IEEE, 2012.

\bibitem{ARF16}
F.~Arrigoni, B.~Rossi, and A.~Fusiello.
\newblock Spectral synchronization of multiple views in {SE(3)}.
\newblock {\em SIAM Journal on Imaging Sciences}, 9(4):1963--1990, 2016.

\bibitem{BGHH18}
C.~Bajaj, T.~Gao, Z.~He, Q.~Huang, and Z.~Liang.
\newblock {SMAC}: simultaneous mapping and clustering using spectral
  decompositions.
\newblock In {\em International Conference on Machine Learning}, pages
  324--333, 2018.

\bibitem{BBS17}
A.~S. Bandeira, N.~Boumal, and A.~Singer.
\newblock Tightness of the maximum likelihood semidefinite relaxation for
  angular synchronization.
\newblock {\em Mathematical Programming}, 163(1-2):145--167, 2017.

\bibitem{BBV16}
A.~S. Bandeira, N.~Boumal, and V.~Voroninski.
\newblock On the low-rank approach for semidefinite programs arising in
  synchronization and community detection.
\newblock In {\em Conference on Learning Theory}, pages 361--382, 2016.

\bibitem{BSS13}
A.~S. Bandeira, A.~Singer, and D.~A. Spielman.
\newblock A {Cheeger} inequality for the graph connection laplacian.
\newblock {\em SIAM Journal on Matrix Analysis and Applications},
  34(4):1611--1630, 2013.

\bibitem{BN11}
F.~Benaych-Georges and R.~R. Nadakuditi.
\newblock The eigenvalues and eigenvectors of finite, low rank perturbations of
  large random matrices.
\newblock {\em Advances in Mathematics}, 227(1):494--521, 2011.

\bibitem{BDS21}
M.~Boedihardjo, S.~Deng, and T.~Strohmer.
\newblock A performance guarantee for spectral clustering.
\newblock {\em SIAM Journal on Mathematics of Data Science}, 3(1):369--387,
  2021.

\bibitem{B15}
N.~Boumal.
\newblock A {R}iemannian low-rank method for optimization over semidefinite
  matrices with block-diagonal constraints.
\newblock {\em arXiv preprint arXiv:1506.00575}, 2015.

\bibitem{B16}
N.~Boumal.
\newblock Nonconvex phase synchronization.
\newblock {\em SIAM Journal on Optimization}, 26(4):2355--2377, 2016.

\bibitem{BVB20}
N.~Boumal, V.~Voroninski, and A.~S. Bandeira.
\newblock Deterministic guarantees for burer-monteiro factorizations of smooth
  semidefinite programs.
\newblock {\em Communications on Pure and Applied Mathematics}, 73(3):581--608,
  2020.

\bibitem{BM05}
S.~Burer and R.~D. Monteiro.
\newblock Local minima and convergence in low-rank semidefinite programming.
\newblock {\em Mathematical Programming}, 103(3):427--444, 2005.

\bibitem{CDF09}
M.~Capitaine, C.~Donati-Martin, and D.~F{\'e}ral.
\newblock The largest eigenvalues of finite rank deformation of large {Wigner}
  matrices: convergence and nonuniversality of the fluctuations.
\newblock {\em The Annals of Probability}, 37(1):1--47, 2009.

\bibitem{CKS15}
K.~N. Chaudhury, Y.~Khoo, and A.~Singer.
\newblock Global registration of multiple point clouds using semidefinite
  programming.
\newblock {\em SIAM Journal on Optimization}, 25(1):468--501, 2015.

\bibitem{CC18}
Y.~Chen and E.~J. Cand{\`e}s.
\newblock The projected power method: An efficient algorithm for joint
  alignment from pairwise differences.
\newblock {\em Communications on Pure and Applied Mathematics},
  71(8):1648--1714, 2018.

\bibitem{CFMW19}
Y.~Chen, J.~Fan, C.~Ma, and K.~Wang.
\newblock Spectral method and regularized {MLE} are both optimal for top-k
  ranking.
\newblock {\em Annals of Statistics}, 47(4):2204, 2019.

\bibitem{CGH14}
Y.~Chen, L.~Guibas, and Q.~Huang.
\newblock Near-optimal joint object matching via convex relaxation.
\newblock {\em Proceedings of the 31st International Conference on Machine
  Learning}, 32(2):100--108, 2014.

\bibitem{CSG16}
Y.~Chen, C.~Suh, and A.~J. Goldsmith.
\newblock Information recovery from pairwise measurements.
\newblock {\em IEEE Transactions on Information Theory}, 62(10):5881--5905,
  2016.

\bibitem{CLS12}
M.~Cucuringu, Y.~Lipman, and A.~Singer.
\newblock Sensor network localization by eigenvector synchronization over the
  {Euclidean} group.
\newblock {\em ACM Transactions on Sensor Networks (TOSN)}, 8(3):1--42, 2012.

\bibitem{DCT21}
A.~d'Aspremont, M.~Cucuringu, and H.~Tyagi.
\newblock Ranking and synchronization from pairwise measurements via {SVD}.
\newblock {\em Journal of Machine Learning Research}, 22:19--1, 2021.

\bibitem{DK70}
C.~Davis and W.~M. Kahan.
\newblock The rotation of eigenvectors by a perturbation. {III}.
\newblock {\em SIAM Journal on Numerical Analysis}, 7(1):1--46, 1970.

\bibitem{DLS21}
S.~Deng, S.~Ling, and T.~Strohmer.
\newblock Strong consistency, graph laplacians, and the stochastic block model.
\newblock {\em Journal of Machine Learning Research}, 22(117):1--44, 2021.

\bibitem{EBW18}
J.~Eldridge, M.~Belkin, and Y.~Wang.
\newblock Unperturbed: spectral analysis beyond {Davis-Kahan}.
\newblock In F.~Janoos, M.~Mohri, and K.~Sridharan, editors, {\em Proceedings
  of Algorithmic Learning Theory}, volume~83 of {\em Proceedings of Machine
  Learning Research}, pages 321--358. PMLR, 07--09 Apr 2018.

\bibitem{FWZ18}
J.~Fan, W.~Wang, and Y.~Zhong.
\newblock An $\ell_{\infty}$ eigenvector perturbation bound and its application
  to robust covariance estimation.
\newblock {\em Journal of Machine Learning Research}, 18(207):1--42, 2018.

\bibitem{GK06}
A.~Giridhar and P.~R. Kumar.
\newblock Distributed clock synchronization over wireless networks: Algorithms
  and analysis.
\newblock In {\em Proceedings of the 45th IEEE Conference on Decision and
  Control}, pages 4915--4920. IEEE, 2006.

\bibitem{G11}
D.~Gross.
\newblock Recovering low-rank matrices from few coefficients in any basis.
\newblock {\em IEEE Transactions on Information Theory}, 57(3):1548--1566,
  2011.

\bibitem{HG13}
Q.-X. Huang and L.~Guibas.
\newblock Consistent shape maps via semidefinite programming.
\newblock {\em Computer Graphics Forum}, 32(5):177--186, 2013.

\bibitem{IPSV20}
M.~A. Iwen, B.~Preskitt, R.~Saab, and A.~Viswanathan.
\newblock Phase retrieval from local measurements: Improved robustness via
  eigenvector-based angular synchronization.
\newblock {\em Applied and Computational Harmonic Analysis}, 48(1):415--444,
  2020.

\bibitem{JCL20}
J.~H. Jung, H.~W. Chung, and J.~O. Lee.
\newblock Weak detection in the spiked wigner model with general rank.
\newblock {\em arXiv preprint arXiv:2001.05676}, 2020.

\bibitem{KX16}
V.~Koltchinskii and D.~Xia.
\newblock Perturbation of linear forms of singular vectors under {G}aussian
  noise.
\newblock In {\em High Dimensional Probability VII}, pages 397--423. Springer,
  2016.

\bibitem{K55}
H.~W. Kuhn.
\newblock The {H}ungarian method for the assignment problem.
\newblock {\em Naval Research Logistics Quarterly}, 2(1-2):83--97, 1955.

\bibitem{LHBC19}
P.-Y. Lajoie, S.~Hu, G.~Beltrame, and L.~Carlone.
\newblock Modeling perceptual aliasing in {SLAM} via discrete--continuous
  graphical models.
\newblock {\em IEEE Robotics and Automation Letters}, 4(2):1232--1239, 2019.

\bibitem{LS19}
G.~Lerman and Y.~Shi.
\newblock Robust group synchronization via cycle-edge message passing.
\newblock {\em arXiv preprint arXiv:1912.11347}, 2019.

\bibitem{L95}
R.-C. Li.
\newblock New perturbation bounds for the unitary polar factor.
\newblock {\em SIAM Journal on Matrix Analysis and Applications},
  16(1):327--332, 1995.

\bibitem{L20b}
S.~Ling.
\newblock Solving orthogonal group synchronization via convex and low-rank
  optimization: tightness and landscape analysis.
\newblock {\em arXiv preprint arXiv:2006.00902}, 2020.

\bibitem{L21a}
S.~Ling.
\newblock Generalized power method for generalized orthogonal {Procrustes}
  problem: Global convergence and optimization landscape analysis.
\newblock {\em arXiv preprint arXiv:2106.15493}, 2021.

\bibitem{L21b}
S.~Ling.
\newblock Near-optimal bounds for generalized orthogonal {P}rocrustes problem
  via generalized power method.
\newblock {\em arXiv preprint arXiv:2112.13725}, 2021.

\bibitem{L20c}
S.~Ling.
\newblock Improved performance guarantees for orthogonal group synchronization
  via generalized power method.
\newblock {\em SIAM Journal on Optimization}, 2022.

\bibitem{LXB19}
S.~Ling, R.~Xu, and A.~S. Bandeira.
\newblock On the landscape of synchronization networks: A perspective from
  nonconvex optimization.
\newblock {\em SIAM Journal on Optimization}, 29(3):1879--1907, 2019.

\bibitem{LYS17}
H.~Liu, M.-C. Yue, and A.~Man-Cho~So.
\newblock On the estimation performance and convergence rate of the generalized
  power method for phase synchronization.
\newblock {\em SIAM Journal on Optimization}, 27(4):2426--2446, 2017.

\bibitem{LYS20}
H.~Liu, M.-C. Yue, and A.~M.-C. So.
\newblock A unified approach to synchronization problems over subgroups of the
  orthogonal group.
\newblock {\em arXiv preprint arXiv:2009.07514}, 2020.

\bibitem{MWCC20}
C.~Ma, K.~Wang, Y.~Chi, and Y.~Chen.
\newblock Implicit regularization in nonconvex statistical estimation: Gradient
  descent converges linearly for phase retrieval, matrix completion, and blind
  deconvolution.
\newblock {\em Foundations of Computational Mathematics}, 20:451--632, 2020.

\bibitem{MMMO17}
S.~Mei, T.~Misiakiewicz, A.~Montanari, and R.~I. Oliveira.
\newblock {Solving SDPs for synchronization and MaxCut problems via the
  Grothendieck inequality}.
\newblock In {\em Conference on Learning Theory}, pages 1476--1515, 2017.

\bibitem{OVW16}
S.~O'Rourke, V.~Vu, and K.~Wang.
\newblock Eigenvectors of random matrices: a survey.
\newblock {\em Journal of Combinatorial Theory, Series A}, 144:361--442, 2016.

\bibitem{OVW18}
S.~O'Rourke, V.~Vu, and K.~Wang.
\newblock Random perturbation of low rank matrices: Improving classical bounds.
\newblock {\em Linear Algebra and its Applications}, 540:26--59, 2018.

\bibitem{OVBS17}
O.~Ozyesil, V.~Voroninski, R.~Basri, and A.~Singer.
\newblock A survey of structure from motion.
\newblock {\em arXiv preprint arXiv:1701.08493}, 2017.

\bibitem{PKSS14}
D.~Pachauri, R.~Kondor, G.~Sargur, and V.~Singh.
\newblock Permutation diffusion maps ({PDM}) with application to the image
  association problem in computer vision.
\newblock In {\em Advances in Neural Information Processing Systems}, pages
  541--549, 2014.

\bibitem{PKS13}
D.~Pachauri, R.~Kondor, and V.~Singh.
\newblock Solving the multi-way matching problem by permutation
  synchronization.
\newblock In {\em Advances in Neural Information Processing Systems}, pages
  1860--1868, 2013.

\bibitem{PWBM18b}
A.~Perry, A.~S. Wein, A.~S. Bandeira, and A.~Moitra.
\newblock Message-passing algorithms for synchronization problems over compact
  groups.
\newblock {\em Communications on Pure and Applied Mathematics},
  71(11):2275--2322, 2018.

\bibitem{PWBM18}
A.~Perry, A.~S. Wein, A.~S. Bandeira, and A.~Moitra.
\newblock Optimality and sub-optimality of {PCA I}: {S}piked random matrix
  models.
\newblock {\em The Annals of Statistics}, 46(5):2416--2451, 2018.

\bibitem{RG20}
E.~Romanov and M.~Gavish.
\newblock The noise-sensitivity phase transition in spectral group
  synchronization over compact groups.
\newblock {\em Applied and Computational Harmonic Analysis}, 49(3):935--970,
  2020.

\bibitem{RCBL19}
D.~M. Rosen, L.~Carlone, A.~S. Bandeira, and J.~J. Leonard.
\newblock {SE-Sync}: A certifiably correct algorithm for synchronization over
  the special {Euclidean} group.
\newblock {\em The International Journal of Robotics Research},
  38(2-3):95--125, 2019.

\bibitem{SHSS16}
Y.~Shen, Q.~Huang, N.~Srebro, and S.~Sanghavi.
\newblock Normalized spectral map synchronization.
\newblock In {\em Advances in Neural Information Processing Systems}, pages
  4925--4933, 2016.

\bibitem{SLL20}
Y.~Shi, S.~Li, and G.~Lerman.
\newblock Robust multi-object matching via iterative reweighting of the graph
  connection {L}aplacian.
\newblock {\em arXiv preprint arXiv:2006.06658}, 2020.

\bibitem{S11}
A.~Singer.
\newblock Angular synchronization by eigenvectors and semidefinite programming.
\newblock {\em Applied and Computational Harmonic Analysis}, 30(1):20--36,
  2011.

\bibitem{S18}
A.~Singer et~al.
\newblock Mathematics for cryo-electron microscopy.
\newblock {\em Proceedings of the International Congress of Mathematicians
  (ICM)}, 3:3981--4000, 2018.

\bibitem{S98}
G.~W. Stewart.
\newblock Perturbation theory for the singular value decomposition.
\newblock Technical Report CS-TR 2539, University of Maryland, 1998.

\bibitem{T12}
J.~A. Tropp.
\newblock User-friendly tail bounds for sums of random matrices.
\newblock {\em Foundations of Computational Mathematics}, 12(4):389--434, 2012.

\bibitem{V18}
R.~Vershynin.
\newblock {\em High-Dimensional Probability: An Introduction with Applications
  in Data Science}, volume~47.
\newblock Cambridge University Press, 2018.

\bibitem{V07}
U.~Von~Luxburg.
\newblock A tutorial on spectral clustering.
\newblock {\em Statistics and Computing}, 17(4):395--416, 2007.

\bibitem{WS13}
L.~Wang and A.~Singer.
\newblock Exact and stable recovery of rotations for robust synchronization.
\newblock {\em Information and Inference: A Journal of the IMA}, 2(2):145--193,
  2013.

\bibitem{Y12}
S.~Yu.
\newblock Angular embedding: A robust quadratic criterion.
\newblock {\em IEEE Transactions on Pattern Analysis and Machine Intelligence},
  34(1):158--173, 2012.

\bibitem{Z19}
T.~Zhang.
\newblock Tightness of the semidefinite relaxation for orthogonal trace-sum
  maximization.
\newblock {\em arXiv preprint arXiv:1911.08700}, 2019.

\bibitem{ZB18}
Y.~Zhong and N.~Boumal.
\newblock Near-optimal bounds for phase synchronization.
\newblock {\em SIAM Journal on Optimization}, 28(2):989--1016, 2018.

\end{thebibliography}
\bibliographystyle{abbrv}

\end{document}